\documentclass[english]{article}
\usepackage[T1]{fontenc}
\usepackage[latin9]{inputenc}
\usepackage{geometry}
\usepackage{amsthm}
\usepackage{amsmath}
\usepackage{amssymb}
\usepackage{fancyhdr}

\makeatletter
\numberwithin{equation}{section}
\numberwithin{figure}{section}
\theoremstyle{plain}
\newtheorem{thm}{\protect\theoremname}
  \theoremstyle{definition}
  
  \theoremstyle{plain}
  \newtheorem{lem}[thm]{\protect\lemmaname}
  \theoremstyle{definition}
  \newtheorem{problem}[thm]{\protect\problemname}
  \theoremstyle{remark}
  
  \theoremstyle{remark}
  \newtheorem*{rem*}{\protect\remarkname}
 \theoremstyle{plain}

\@ifundefined{date}{}{\date{}}
\usepackage{tikz} 

\input{Qcircuit}
\usepackage{graphicx,amsmath, amsthm, amssymb,color,url,booktabs,comment}  
\usepackage[colorlinks]{hyperref} 

\newcommand{\nc}{\newcommand}
\nc{\rnc}{\renewcommand}

\newcommand{\braket}[2]{\langle #1|#2\rangle}

\nc{\vev}[1]{\langle#1\rangle}
\nc{\grad}{{\vec{\nabla}}}

\DeclareMathOperator{\tr}{tr}

\newcommand{\be}{\begin{equation}}
\newcommand{\ee}{\end{equation}}
\newcommand{\bea}{\begin{eqnarray}}
\newcommand{\eea}{\end{eqnarray}}
\newcommand{\nn}{\nonumber}
\newcommand{\bi}{\begin{itemize}}
\newcommand{\ei}{\end{itemize}}
\newcommand{\bn}{\begin{enumerate}}
\newcommand{\en}{\end{enumerate}}
\def\beas#1\eeas{\begin{eqnarray*}#1\end{eqnarray*}}
\def\ba#1\ea{\begin{align}#1\end{align}}
\nc{\bas}{\[\begin{aligned}}
\nc{\eas}{\end{aligned}\]}
\nc{\bpm}{\begin{pmatrix}}
\nc{\epm}{\end{pmatrix}}

\def\nn{\nonumber}

\def\L{\left} 
\def\R{\right}

\newtheorem{cor}[thm]{Corollary}

\makeatletter
\newtheorem*{rep@theorem}{\rep@title}
\newcommand{\newreptheorem}[2]{%
\newenvironment{rep#1}[1]{%
 \def\rep@title{#2 \ref{##1}}%
 \begin{rep@theorem}}%
 {\end{rep@theorem}}}
\makeatother

\newreptheorem{thm}{Theorem}
\newreptheorem{lem}{Lemma}

\def\benum{\begin{enumerate}}
\def\eenum{\end{enumerate}}
\def\bit{\begin{itemize}}
\def\eit{\end{itemize}}
\def\bdesc{\begin{description}}
\def\edesc{\end{description}}

\nc{\todo}[1]{\textcolor{red}{todo: #1}}

\def\begsub#1#2\endsub{\begin{subequations}\label{eq:#1}\begin{align}#2\end{align}\end{subequations}}
\nc\qand{\qquad\text{and}\qquad}
\nc\mnb[1]{\medskip\noindent{\bf #1}}
\nc\mn{\medskip\noindent}

\nc{\nl}{\nn \\ &=}  
\nc{\nnl}{\nn \\ &}  
\nc{\fot}{\frac{1}{2}} 
\newcommand{\ben}{\begin{enumerate}}
\newcommand{\een}{\end{enumerate}}
\nc{\mc}{\mathcal}
\nc{\beq}{\begin{equation}}
\nc{\eeq}{\end{equation}}
\makeatother

\usepackage{babel}
  \providecommand{\algorithmname}{Algorithm}
  \providecommand{\definitionname}{Definition}
  \providecommand{\lemmaname}{Lemma}
  \providecommand{\problemname}{Problem}
  \providecommand{\remarkname}{Remark}
\providecommand{\theoremname}{Theorem}

\theoremstyle{plain}
\newtheorem{obs}[thm]{Observation}

\theoremstyle{definition}
\newtheorem{alg}[thm]{\protect\algorithmname}

\theoremstyle{remark}
\newtheorem*{alg_in}{Input}

\newcommand{\Expect}{{\rm I\kern-.3em E}}

\usepackage{authblk}
\author[1]{\normalsize Cedric Yen-Yu Lin\thanks{cedricl@mit.edu}}
\author[1]{\normalsize Han-Hsuan Lin\thanks{hanmas@mit.edu}}
\affil[1]{\small \emph{Center for Theoretical Physics, Massachusetts Institute of Technology}}

\title{Upper bounds on quantum query complexity inspired by the Elitzur-Vaidman bomb tester}
\date{ }

\begin{document}


\maketitle

\begin{abstract}
Inspired by the Elitzur-Vaidman bomb testing problem \cite{elitzur-1993}, we introduce a new query complexity model, which we call bomb query complexity $B(f)$. We investigate its relationship with the usual quantum query complexity $Q(f)$, and show that $B(f)=\Theta(Q(f)^2)$. 

This result gives a new method to upper bound the quantum query complexity: we give a method of finding bomb query algorithms from classical algorithms, which then provide nonconstructive upper bounds on $Q(f)=\Theta(\sqrt{B(f)})$. We subsequently were able to give explicit quantum algorithms matching our upper bound method.  We apply this method on the single-source shortest paths problem on unweighted graphs, obtaining an algorithm with $O(n^{1.5})$ quantum query complexity, improving the best known algorithm of $O(n^{1.5}\sqrt{\log n})$ \cite{furrow-2008}. Applying this method to the maximum bipartite matching problem gives an $O(n^{1.75})$ algorithm, improving the best known trivial $O(n^2)$ upper bound.
\end{abstract}

\thispagestyle{fancy}
\rhead{MIT-CTP/4592}
\renewcommand{\headrulewidth}{0pt}
\renewcommand{\footrulewidth}{0pt}

\section{Introduction}
Quantum query complexity is an important method of understanding the power of quantum computers. In this model we are given a black-box containing a boolean string $x=x_1\cdots x_N$, and we would like to calculate some function $f(x)$ with as few quantum accesses to the black-box as possible. It is often easier to give bounds on the query complexity than to the time complexity of a problem, and insights from the former often prove useful in understanding the power and limitations of quantum computers. One famous example is Grover's algorithm for unstructured search \cite{grover-1997}; by casting this problem into the query model it was shown that $\Theta(\sqrt{N})$ queries was required \cite{bbbv-1997}, proving that Grover's algorithm is optimal.

Several methods have been proposed to bound the quantum query complexity. Upper bounds are almost always proven by finding better query algorithms. Some general methods of constructing quantum algorithms have been proposed, such as quantum walks \cite{ambainis-2003,szegedy-2004,magniez-2006,jeffery-2012} and learning graphs \cite{belovs-2011}. For lower bounds, the main methods are the polynomial method \cite{beals-1998} and adversary method \cite{ambainis-2000}. In particular, the general adversary lower bound \cite{hoyer-2007} has been shown to tightly characterize quantum query complexity \cite{reichardt-2009,reichardt-2011,lee-2010}, but calculating such a tight bound seems difficult in general. Nevertheless, the general adversary lower bound is valuable as a theoretical tool, for example in proving composition theorems \cite{reichardt-2011,lee-2010,kimmel-2013} or showing nonconstructive (!) upper bounds \cite{kimmel-2013}.

\subsection*{Our work}
To improve our understanding of quantum query complexity, we introduce and study an alternative oracle model, which we call the \emph{bomb oracle} (see Section \ref{sect:bomb query} for the precise definition). Our model is inspired by the concept of \emph{interaction free measurements}, illustrated vividly by the Elitzur-Vaidman bomb testing problem \cite{elitzur-1993}, in which a property of a system can be measured without disturbing the system significantly. Like the quantum oracle model, in the bomb oracle model we want to evaluate a function $f(x)$ on a hidden boolean string $x=x_1\cdots x_N$ while querying the oracle as few times as possible. In this model, however, the bomb oracle is a controlled quantum oracle with the extra requirement that the algorithm fails if the controlled query returns a 1. This seemingly impossible task can be tackled using the quantum Zeno effect \cite{misra-1977}, in a fashion similar to the Elitzur-Vaidman bomb tester \cite{kwiat-1995} (Section \ref{sect:elitzur-vaidman}).

Our main result (Theorem \ref{thm:main result}) is that the bomb query complexity, $B(f)$, is characterized by the square of the quantum query complexity $Q(f)$: 

\begin{repthm}{thm:main result}
\begin{equation}
B(f) = \Theta(Q(f)^2).
\end{equation}
\end{repthm}

We prove the upper bound, $B(f) = O(Q(f)^2)$ (Theorem \ref{thm:upperbound}), by adapting Kwiat et al.'s solution of the Elitzur-Vaidman bomb testing problem (Section \ref{sect:elitzur-vaidman}, \cite{kwiat-1995}) to our model. We prove the lower bound, $B(f) = \Omega(Q(f)^2)$ (Theorem \ref{thm:lowerbound}), by demonstrating that $B(f)$ is lower bounded by the square of the general adversary bound \cite{hoyer-2007}, $(\text{Adv}^\pm(f))^2$. The aforementioned result that the general adversary bound tightly characterizes the quantum query complexity \cite{reichardt-2009,reichardt-2011,lee-2010}, $Q(f)= \Theta(\text{Adv}^\pm(f) )$, allows us to draw our conclusion.

This characterization of Theorem  \ref{thm:main result} allows us to give \emph{nonconstructive} upper bounds to the quantum query complexity for some problems. For some functions $f$ a bomb query algorithm is easily designed by adapting a classical algorithm: specifically, we show that (stated informally):

\begin{repthm}{thm:classical}[informal]
Suppose there is a classical algorithm that computes $f(x)$ in $T$ queries, and the algorithm guesses the result of each query (0 or 1), making no more than an expected $G$ mistakes for all $x$. Then we can design a bomb query algorithm that uses $O(TG)$ queries, and hence $B(f) = O(TG)$. By our characterization of Theorem \ref{thm:main result}, $Q(f)=O(\sqrt{TG})$.
\end{repthm}

This result inspired us to look for an explicit quantum algorithm that reproduces the query complexity $O(\sqrt{TG})$. We were able to do so:

\begin{repthm}{thm:classical to quantum}
Under the assumptions of Theorem \ref{thm:classical}, there is an explicit algorithm (Algorithm \ref{alg:classical}) for $f$ with query complexity $O(\sqrt{TG})$.
\end{repthm}

Using Algorithm \ref{alg:classical}, we were able to give an $O(n^{3/2})$ algorithm for the single-source shortest paths (SSSP) problem in an unweighted graph with $n$ vertices, beating the best-known $O(n^{3/2} \sqrt{\log n})$ algorithm \cite{furrow-2008}. A more striking application is our $O(n^{7/4})$ algorithm for maximum bipartite matching; in this case the best-known upper bound was the trivial $O(n^2)$, although the time complexity of this problem had been studied in \cite{ambainis-2006} (and similar problems in \cite{dorn-2008}).

Finally, in Section \ref{sect:projective query} we briefly discuss a related query complexity model, which we call the \emph{projective query complexity} $P(f)$,  in which each quantum query to $x$ is immediately followed by a classical measurement of the query result. This model seems interesting to us because its power lies between classical and quantum: we observe that $P(f) \le B(f) = \Theta(Q(f)^2)$ and $Q(f) \le P(f) \le R(f)$, where $R(f)$ is the classical randomized query complexity. We note that Regev and Schiff \cite{regev-2008} showed that $P(OR)=\Theta(N)$.

\subsection*{Past and related work}
Mitchison and Jozsa have proposed a different computational model called \emph{counterfactual computation} \cite{mitchison-2001}, also based on interaction-free measurement. In counterfactual computation the result of a computation may be learnt without ever running the computer. There has been some discussion on what constitutes counterfactual computation; see for example \cite{hosten-2006,mitchison-2006,hosten-2006-2,vaidman-2006,hosten-2006-3,salih-2013,vaidman-2013}.

There have also been other applications of interaction-free measurement to quantum cryptography. For example, Noh has proposed counterfactual quantum cryptography \cite{noh-2009}, where a secret key is distributed between parties, even though a particle carrying secret information is not actually transmitted. More recently, Brodutch et al. proposed an adaptive attack \cite{brodutch-2014} on Wiesner's quantum money scheme \cite{wiesner-1983}; this attack is directly based off Kwiat et al.'s solution of the Elitzur-Vaidman bomb testing problem \cite{kwiat-1995}.

Our Algorithm \ref{alg:classical} is very similar to Kothari's algorithm for the oracle identification problem \cite{kothari-2014}, and we refer to his analysis of the query complexity in our work. 

The projective query model we detail in Section \ref{sect:projective query} was, to our knowledge, first considered by Aaronson in unpublished work in 2002 \cite{aaronson-2014}.

\subsection*{Discussion and outlook}
Our work raises a number of open questions. The most obvious ones are those pertaining to the application of our recipe for turning classical algorithms into bomb algorithms, Theorem \ref{thm:classical}:

\begin{itemize}
\item Can we generalize our method to handle non-boolean input and output? If so, we might be able to find better upper bounds for the adjacency-list model, or to study graph problems with weighted edges.
\item Can our explicit (through Theorem \ref{thm:classical to quantum}) algorithm for maximum bipartite matching be made more \emph{time} efficient? The best known quantum algorithm for this task has time complexity $O(n^2\log n)$ in the adjacency matrix model \cite{ambainis-2006}.
\item Finally, can we find more upper bounds using Theorem \ref{thm:classical}? For example, could the query complexity of the maximum matching problem on general nonbipartite graphs be improved with Theorem \ref{thm:classical}, by analyzing the classical algorithm of Micali and Vazirani \cite{micali-1980}?
\end{itemize}

Perhaps more fundamental, however, is the possibility that the bomb query complexity model will help us understand the relationship between the classical randomized query complexity, $R(f)$, and the quantum query complexity $Q(f)$. It is known \cite{beals-1998} that for all total functions $f$, $R(f) = O(Q(f)^6)$; however, there is a long-standing conjecture that actually $R(f) = O(Q(f)^2)$. In light of our results, this conjecture is equivalent to the conjecture that $R(f) = O(B(f))$. Some more open questions, then, are the following:

\begin{itemize}
\item Can we say something about the relationship between $R(f)$ and $B(f)$ for specific classes of functions? For example, is $R(f) = O(B(f)^2)$ for total functions?
\item Referring to the notation of Theorem \ref{thm:classical}, we have $B(f) = O(TG)$. Is the quantity $G$ related to other measures used in the study of classical decision-tree complexity, for example the certificate complexity, sensitivity \cite{cook-1986}, block sensitivity \cite{nisan-1991}, or (exact or approximate) polynomial degree? (For a review, see \cite{buhrman-1999}.)
\item What about other query complexity models that might help us understand the relationship between $R(f)$ and $Q(f)$? One possibility is the projective query complexity, $P(f)$, considered in Section \ref{sect:projective query}. Regev and Schiff \cite{regev-2008} have shown (as a special case of their results) that even with such an oracle, $P(OR)=\Theta(N)$ queries are needed to evaluate the OR function.
\end{itemize}

We hope that further study on the relationship between bomb and classical randomized complexity will shed light on the power and limitations of quantum computation.
\section{Preliminaries}

\subsection{The Elitzur-Vaidman bomb testing problem} \label{sect:elitzur-vaidman}

The Elitzur-Vaidman bomb testing problem \cite{elitzur-1993} is a well-known thought experiment to demonstrate how quantum mechanics differs drastically from our classical perceptions. This problem demonstrates dramatically the possibility of \emph{interaction free measurements}, the possibility of a measurement on a property of a system without disturbing the system. 

The bomb-testing problem is as follows: assume we have a bomb that is either a dud or a live bomb. The only way to interact with the bomb is to probe it with a photon: if the bomb is a dud, then the photon passes through unimpeded; if the bomb is live, then the bomb explodes. We would like to determine whether the bomb is live or not without exploding it. If we pass the photon through a beamsplitter before probing the bomb, we can implement the \emph{controlled probe}, pictured below:

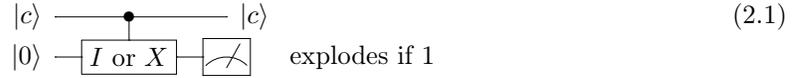
\begin{align} \label{circ:bombEV}
\Qcircuit @C=1em @R=.7em {
\lstick{\ket{c}}& \ctrl{1} & \rstick{\ket{c}} \qw \\
\lstick{\ket{0}}&   \gate{I\text{ or }X} &\meter & \rstick{\text{explodes if 1}}  
}
\end{align}

The controlled gate is $I$ if the bomb is a dud, and $X$ if it is live. It was shown in \cite{kwiat-1995} how to determine whether a bomb was live with arbitrarily low probability of explosion by making use of the quantum Zeno effect \cite{misra-1977}. Specifically, writing $R(\theta)=\exp(i\theta X)$ (the unitary operator rotating $\ket{0}$ to $\ket{1}$ in $\pi/(2\theta)$ steps), the following circuit determines whether the bomb is live with failure probability $O(\theta)$:

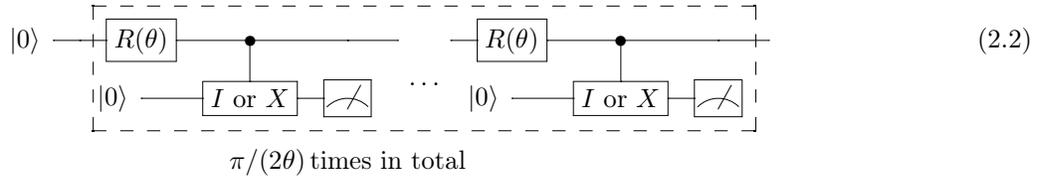
\begin{align}
\Qcircuit @C=1em @R=.7em {
\lstick{\ket{0}}   & \qw &\gate{R(\theta)}& \ctrl{1}  & \qw & \qw &    &&\gate{R(\theta)} & \ctrl{1}  & \qw & \qw \\
&  &\lstick{\ket{0}} &\gate{I\text{ or }X}  & \meter  & & \ustick{\dots} &    &\lstick{\ket{0}}&\gate{I\text{ or }X} & \meter &  \\ 
&&&&\dstick{{\pi/(2\theta)}\,\text{times in total}}  \gategroup{1}{3}{2}{11}{1.0em}{--} \\
}  \\
 \, \nonumber 
\end{align} 

If the bomb is a dud, then the controlled probes do nothing, and repeated application of $R(\theta)$ rotates the control bit from $\ket{0}$ to $\ket{1}$. If the bomb is live, the bomb explodes with $O(\theta^2)$ probability in each application of the probe, projecting the control bit back to $\ket{0}$. After $O(1/\theta)$ iterations the control bit stays in $\ket{0}$, with only a $O(\theta)$ probability of explosion. Using $O(1/\theta)$ operations, we can thus tell a dud bomb apart from a live one with only $O(\theta)$ probability of explosion.

\subsection{Quantum query complexity} \label{sect:quantum query}
Throughout this paper, all functions $f$ which we would like to calculate are assumed to have boolean input, i.e. the domain is $D \subseteq \{0,1\}^N$.

For a boolean strings $x \in \{0,1\}^N$, the quantum oracle $O_x$ is a unitary operator that acts on a one-qubit record register and an $N$-dimensional index register as follows ($\oplus$ is the XOR function):

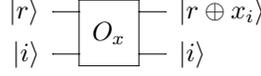
\begin{equation}
O_x\ket{r,i} = \ket{r\oplus x_i,i}
\end{equation}
\begin{align}
\Qcircuit @C=1em @R=.7em {
\lstick{\ket{r}}&\multigate{1}{O_x}   &  \rstick{\ket{r\oplus x_i}}\qw\\
\lstick{\ket{i}}&   \ghost{O_x}  & \rstick{\ket{i}} \qw }    \nn
\end{align}

We want to determine the value of a boolean function $f(x)$ using as few queries to the quantum oracle $O_x$ as possible. Algorithms for $f$ have the general form as the following circuit, where the $U_t$'s are unitaries independent of $x$:  
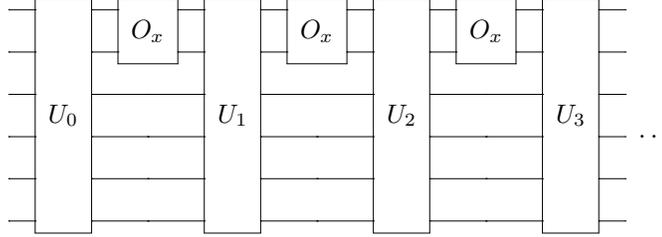
$$
\Qcircuit @C=1em @R=.7em {
 &\multigate{5}{U_0}   & \multigate{1}{O_x}  &\multigate{5}{U_1} & \multigate{1}{O_x} &\multigate{5}{U_2}  & \multigate{1}{O_x} &\multigate{5}{U_3}&\qw                        \\
 &\ghost{U_0}   &\ghost{O_x}           &\ghost{U_1}   &\ghost{O_x}  &\ghost{U_0}   &\ghost{O_x}  &\ghost{U_0}&\qw    \\
&\ghost{U_0}&   \qw                               & \ghost{U_1} &\qw & \ghost{U_1} &\qw  & \ghost{U_1} &\qw \\
&\ghost{U_0}&   \qw                               & \ghost{U_1} &\qw & \ghost{U_1} &\qw  & \ghost{U_1} & \rstick{\dots}\qw \\
&\ghost{U_0}&   \qw                               & \ghost{U_1} &\qw & \ghost{U_1} &\qw  & \ghost{U_1} &\qw \\
&\ghost{U_0}&   \qw                               & \ghost{U_1} &\qw & \ghost{U_1} &\qw  & \ghost{U_1} &\qw 
}
$$
The quantum query complexity $Q_{\delta}(f)$ is the minumum number of applications of $O_x$'s in the circuit requried to determine $f(x)$ with error no more than $\delta$ for all $x$. By gap amplification (e.g. by performing the circuit multiple rounds and doing majority voting), it can be shown that the choice of $\delta$ only affects the query complexity by a $\log(1/\delta)$ factor. We therefore often set $\delta = 0.01$ and write $Q_{0.01}(f)$ as $Q(f)$.

\section{Bomb query complexity} \label{sect:bomb query}

In this section we introduce a new query complexity model, which we call the \emph{bomb query complexity}. A circuit in the bomb query model is a restricted quantum query circuit, with the following restrictions on the usage of the quantum oracle:

\begin{enumerate}
\item We have an extra control register $\ket{c}$ used to control whether $O_x$ is applied (we call the controlled version $CO_x$):
\begin{equation}
CO_x\ket{c,r,i} = \ket{c,r\oplus(c\cdot x_i),i}.
\end{equation}
where $\cdot$ indicates boolean AND.
\item The record register, $\ket{r}$ in the definition of $CO_x$ above, \emph{must} contain $\ket{0}$ before $CO_x$ is applied.
\item After $CO_x$ is applied, the record register is immediately measured in the computational basis (giving the answer $c\cdot x_i$), and the algorithm \emph{terminates immediately if a 1 is measured} (if $c\cdot x_i = 1$). We refer to this as \emph{the bomb blowing up} or \emph{the bomb exploding}.
\end{enumerate}

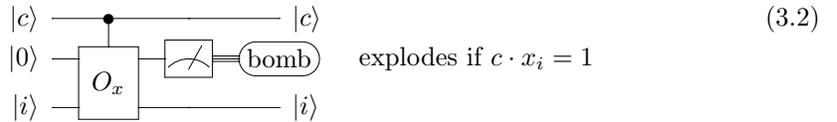
\begin{align} \label{circ:bomb}
\Qcircuit @C=1em @R=.7em {
\lstick{\ket{c}}& \ctrl{1} &\qw & \rstick{\ket{c}} \qw \\
\lstick{\ket{0}}&   \multigate{1}{O_x} &\meter & \measure{\text{bomb}}\cw & \rstick{\text{explodes if }c\cdot x_i =1}     \\
\lstick{\ket{i}}&\ghost{O_x}  &\qw & \rstick{\ket{i}} \qw
}
\end{align}

We define the \emph{bomb query complexity} $B_{\epsilon,\delta}(f)$ to be the minimum number of times the above circuit needs to be applied in an algorithm such that the following hold for all input $x$:
\begin{itemize}
\item The algorithm reaches the end without the bomb exploding with probability at least  $1-\epsilon$. We refer to the probability that the bomb explodes as the \emph{probability of explosion}.
\item The total probability that the bomb either explodes or fails to output $f(x)$ correctly is no more than $\delta \ge \epsilon$.
\end{itemize}
The above implies that the algorithm outputs the correct answer with probability at least $1-\delta$.

The effect of the above circuit is equivalent to applying the following projector on $\ket{c,i}$:
\begin{align}
M_x = CP_{x,0} &= \sum\limits_{i=1}^{N}\ket{0,i}\bra{0,i} + \sum\limits_{x_i=0}\ket{1,i} \bra{1,i} \\
&= I - \sum\limits_{x_i=1}\ket{1,i} \bra{1,i}.
\end{align}

$CP_{x,0}$ (which we will just call $M_x$ in our proofs later on) is the controlled version of $P_{x,0}$, the projector that projects onto the indices $i$ on which $x_i = 0$:
\begin{align}
P_{x,0} = \sum\limits_{x_i=0}\ket{i}\bra{i}.
\end{align}

Thus Circuit \ref{circ:bomb} is equivalent to the following circuit :
\begin{align}
&\Qcircuit @C=1em @R=.7em {
\lstick{\ket{c}}& \ctrl{1} &\qw & \rstick{\ket{c}} \qw \\
\lstick{\ket{i}}&\gate{P_{x,0}}  &\qw & \rstick{(1-c \cdot x_i)\ket{i}} \qw
}
\end{align}
In this notation, the square of the norm of a state is the probability that the state has survived to this stage, i.e. the algorithm has not terminated. The norm of $(1-c \cdot x_i)\ket{x_i}$ is 1 if $c \cdot x_i = 0$ (the state survives this stage), and 0 otherwise (the bomb blows up).

A general circuit in this model looks like the following:

$$
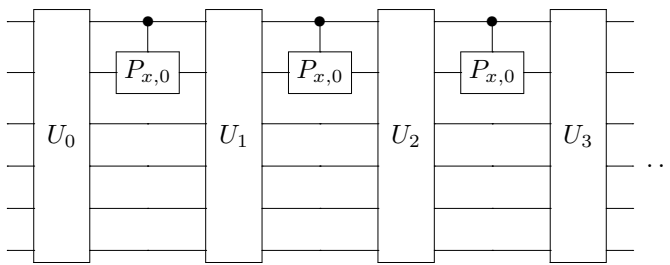
\Qcircuit @C=1em @R=.7em {
 &\multigate{5}{U_0}   & \ctrl{1}  &\multigate{5}{U_1} & \ctrl{1} &\multigate{5}{U_2}  & \ctrl{1} &\multigate{5}{U_3}&\qw                        \\
 &\ghost{U_0}   &\gate{P_{x,0}}           &\ghost{U_1}   &\gate{P_{x,0}}  &\ghost{U_0}   &\gate{P_{x,0}}  &\ghost{U_0}&\qw    \\
&\ghost{U_0}&   \qw                               & \ghost{U_1} &\qw & \ghost{U_1} &\qw  & \ghost{U_1} &\qw \\
&\ghost{U_0}&   \qw                               & \ghost{U_1} &\qw & \ghost{U_1} &\qw  & \ghost{U_1} & \rstick{\dots}\qw \\
&\ghost{U_0}&   \qw                               & \ghost{U_1} &\qw & \ghost{U_1} &\qw  & \ghost{U_1} &\qw \\
&\ghost{U_0}&   \qw                               & \ghost{U_1} &\qw & \ghost{U_1} &\qw  & \ghost{U_1} &\qw 
}
$$

It is not at all clear that gap amplification can be done efficiently in the bomb query model to improve the error $\delta$; after all, repeating the circuit multiple times increases the chance that the bomb blows up.  However, it turns out that the complexity $B_{\epsilon,\delta}(f)$ is closely related to $Q_{\delta}(f)$, and therefore the choice of $\delta$ affects $B_{\epsilon,\delta}(f)$ by at most a $\log^2(1/\delta)$ factor as long as $\delta \ge \epsilon$ (see Lemma \ref{lem:delta}). We therefore often omit $\delta$ by setting $\delta = 0.01$, and write $B_{\epsilon,0.01}(f)$ as $B_{\epsilon}(f)$. Sometimes we even omit the $\epsilon$.

Finally, note that the definition of the bomb query complexity $B(f)$ is inherently \emph{asymmetric} with respect to $0$ and $1$ in the input: querying $1$ causes the bomb to blow up, while querying $0$ is safe. In Section \ref{sect:generalize bomb}, we define a \emph{symmetric}  bomb query model and its corresponding query complexity, $\tilde{B}_{\epsilon,\delta}(f)$.  We prove that this generalized symmetric model is asymptotically equivalent to the original asymmetric model, $\tilde B_{\epsilon,\delta}(f) = \Theta({B}_{\epsilon,\delta}(f))$, in Lemma \ref{lem:equivalent}. This symmetric version of the bomb query complexity will turn out to be useful in designing bomb query algorithms.

\section{Main result}
Our main result is the following:
\begin{thm} \label{thm:main result} For all functions $f$ with boolean input alphabet, and numbers $\epsilon$ satisfying $0 < \epsilon \le 0.01$,

\begin{equation}
B_{\epsilon,0.01}(f) = \Theta\L(\frac{Q_{0.01}(f)^2}{\epsilon}\R).
\end{equation}
Here 0.01 can be replaced by any constant no more than $1/10$.
\end{thm}
\begin{proof}
The upper bound $B_{\epsilon,\delta}(f) = O(Q_\delta(f)^2/\epsilon)$ is proved in Theorem \ref{thm:upperbound}. The lower bound $B_{\epsilon,\delta}(f) = \Omega(Q_{0.01}(f)^2/\epsilon)$ is proved in Theorem \ref{thm:lowerbound}.
\end{proof}
\begin{lem} \label{lem:delta}  For all functions $f$ with boolean input alphabet, and numbers $\epsilon$, $\delta$ satisfying $0 < \epsilon \le \delta \le 1/10$, 
\begin{align}
B_{\epsilon,0.1}(f) = O(B_{\epsilon,\delta}(f)),\quad B_{\epsilon,\delta}(f)= O(B_{\epsilon,0.1}(f)\log^2(1/\delta)).
\end{align}
In particular, if $\delta$ is constant,
\begin{align}
B_{\epsilon,\delta}(f)=\Theta (B_{\epsilon,0.1}(f)).
\end{align}
\end{lem}
\begin{proof}
This follows from Theorem \ref{thm:upperbound} and the fact that $Q_{0.1}(f)= O(Q_{\delta}(f))$ and $Q_{\delta}(f)= O(Q_{0.1}(f)\log(1/\delta))$.
\end{proof}
Because of this result, we will often omit the $0.01$ in $B_{\epsilon,0.01}$ and write simply $B_{\epsilon}$.

\subsection{Upper bound}
\begin{thm} \label{thm:upperbound}
For all functions $f$ with boolean input alphabet, and numbers $\epsilon$, $\delta$ satisfying $0 < \epsilon \le \delta \le 1/10$, 
\begin{equation}
B_{\epsilon,\delta}(f) = O(Q_{\delta}(f)^2 / \epsilon).
\end{equation}
\end{thm}
The proof follows the solution of Elitzur-Vaidman bomb-testing problem (\cite{kwiat-1995}, or Section \ref{sect:elitzur-vaidman}). By taking advantage of the Quantum Zeno effect \cite{misra-1977}, using $O({Q(f) \over \epsilon})$ calls to $M_x$, we can simulate one call to $O_x$ with probabilty of explosion $O({\epsilon \over Q(f)})$. Replacing all $O_x$ queries with this construction results in a bounded error algorithm with probability of explosion $O({\epsilon \over Q(f) } Q(f)) =O(\epsilon)$.
\begin{proof} Let $\theta = \pi/(2L)$ for some large positive integer $L$ (chosen later), and let $R(\theta)$ be the rotation 
\begin{align}
\bpm
\cos\theta &-\sin \theta \\
\sin\theta &\cos\theta
\epm
\end{align}
We claim that with $2L$ calls to the bomb oracle $M_x = CP_{x,0}$, we can simulate $O_x$ by the following circuit with probability of explosion less than $\pi^2/(2L)$ and error $O(1/L)$. 
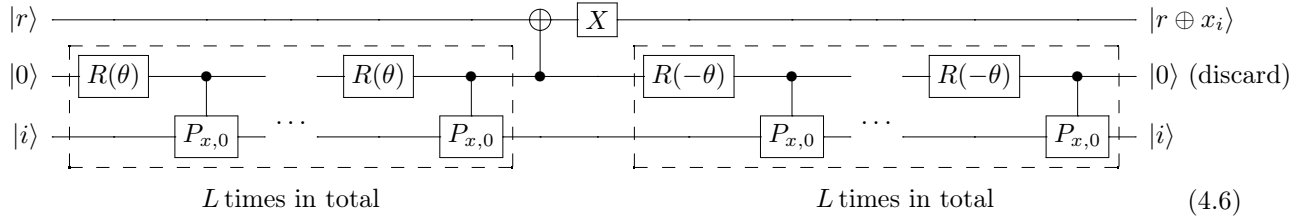
\begin{align}
\Qcircuit @C=1em @R=.7em {
\lstick{\ket{r}} &\qw &\qw&\qw&\qw&\qw&\qw&\qw &\targ &\gate{X} &\qw &\qw&\qw&\qw&\qw&\qw&\qw&\rstick{ \ket{r\oplus x_i}}\qw \\
\lstick{\ket{0}}   &\gate{R(\theta)}& \ctrl{1}  &\qw&    &&\gate{R(\theta)} & \ctrl{1} 
&  \ctrl{-1}  &\qw &\gate{R(-\theta)}& \ctrl{1} &\qw&    &&\gate{R(-\theta)} & \ctrl{1} &\rstick{\ket{0}\text{ (discard)}}\qw   \\
\lstick{\ket{i}}  &\qw &\gate{P_{x,0}}  &\qw    & \ustick{\dots} &    &\qw&\gate{P_{x,0}}  & \qw &\qw
& \qw&\gate{P_{x,0}}   &\qw &    \ustick{\dots} &   &\qw&\gate{P_{x,0}} &\rstick{\ket{i}}\qw\\ 
&&&&\dstick{{L}\,\text{times in total}} &&&&&&&&&&\dstick{\,{L}\,\text{times in total}}  \gategroup{2}{2}{3}{8}{.7em}{--}  \gategroup{2}{11}{3}{17}{.7em}{--} \\
} \nn \\
 \, \label{eq:B2O}
\end{align} 
In words, we simulate $O_x$ acting on $\ket{r,i}$ by the following steps:
\begin{enumerate}
\item Append an ancilla qubit $\ket{0}$, changing the state into $\ket{r,0,i}$.
\item Repeat the following $L$ times: \label{step:zeeno}
\begin{enumerate}
\item apply $R(\theta)$ on the second register
\item apply $M_x$ on the third register controlled by the second register.
\end{enumerate}
At this point, if the bomb hasn't blown up, the second register should contain $1-x_i$.
\item Apply $CNOT$ on the first register controlled by the second register; this copies $1-x_i$ to the first register.
\item Apply a $NOT$ gate to the first register. 
\item Repeat the following $L$ times to uncompute the second (ancilla) register \label{step:uncomp}:
\begin{enumerate}
\item apply $R(-\theta)$ on the second register
\item apply $M_x$ on the third register controlled by second register
\end{enumerate}
\item Discard the second (ancilla) register.
\end{enumerate}	

We now calculate explicitly the action of the circuit on an arbitrary state to confirm our claims above. Consider how the circuit acts on the basis state $\ket{r,0,i}$ (the second register being the appended ancilla). We break into cases:

\begin{itemize}
\item If $x_i=0$, then $P_{x,0}\ket{i} = \ket{i}$, so the controlled projections do nothing. Thus in Step 2 the rotation $R(\theta)^L = R(\pi/2)$ is applied to the ancilla qubit, rotating it from $0$ to $1$. After Step 2 then, the state is $\ket{r,1,i}$. Step 3 and 4 together do not change the state, while Step 5 rotates the ancilla back to $0$, resulting in the final state $\ket{r,0,i}$.
\item If $x_i=1$, then $P_{x,0}\ket{i} = 0$, and
\begin{equation}
M_x\ket{0,i} = \ket{0,i}, \quad M_x\ket{1,i} = 0 \quad \text{(for } x_i = 1 \text{)}
\end{equation}
Therefore in Step 2 and Step 5, after each rotation $R(\pm\theta)$, the projection $CP_{x,0}$ projects the ancilla back to 0:
\begin{equation}
M_xR(\theta)\ket{0,i} = M_x(\cos\theta\ket{0} + \sin\theta\ket{1})\ket{i} = \cos\theta\ket{0,i} \quad \text{(for } x_i = 1 \text{)}
\end{equation}
Each application of $M_xR(\theta)$ thus has no change on the state other than to shrink its amplitude by $\cos\theta$. The CNOT in Step 3 has no effect (since the ancilla stays in 0), and Step 4 maps $\ket{r}$ to $\ket{r\oplus 1}$. Since there are $2L$ applications of this shrinkage (in Step 2 and 5), the final state is $\cos^{2L}\theta \ket{r\oplus 1,0,i}$.
\end{itemize}

We can now combine the two cases: by linearity, the application of the circuit on a general state $\sum_{r,i}a_{r,i}\ket{r,i}$ (removing the ancilla) is
\begin{align}
\sum\limits_{r,i}a_{r,i}\ket{r,i} &\rightarrow \sum\limits_{r \in \{0,1\}, x_i = 0}a_{r,i}\ket{r,i} + \sum\limits_{r \in \{0,1\}, x_i = 1}a_{r,i}\cos^{2L}(\theta)\ket{r\oplus 1,i} \\
&= \sum\limits_{r,i}a_{r,i}\cos^{2Lx_i}\left(\frac{\pi}{2L}\right)\ket{r \oplus x_i,i} \equiv \ket{\psi'}
\end{align}

Thus the effect of this construction simulates the usual quantum oracle $\ket{r,i} \rightarrow \ket{r \oplus x_i,i}$ with blowing up probability no more than
\begin{equation}
1 - \cos^{4L}\left(\frac{\pi}{2L}\right) \le 1 - \left(1 - \frac{\pi^2}{4L^2}\right)^{2L} \le \frac{\pi^2}{2L}.
\end{equation}
Moreover, the difference between the output of our circuit, $\ket{\psi'}$, and the output on the quantum oracle, $\ket{\psi}=\sum_{r,i}a_{r,i}\ket{r\oplus x_i,i}$, is
\begin{align}
\L\| \ket{\psi'} - \ket{\psi} \R\| &= \L\| \sum\limits_{r \in \{0,1\}, x_i = 1}a_{r,i}(1-\cos^{2L}(\theta))\ket{r\oplus 1,i} \R\| \\
&\le 1 - \cos^{2L}\frac{\pi}{2L} \le \frac{\pi^2}{4L}.
\end{align}

Given this construction, we can now prove our theorem. Suppose we are given a quantum algorithm that finds $f(x)$ with $Q_{\delta'}(f)$ queries, making at most $\delta' = \delta - \epsilon$ error. We construct an algorithm using bomb oracles instead by replacing each of the applications of the quantum oracle $O_x$ by our circuit construction (\ref{eq:B2O}), where we choose 
\begin{equation}
L = \L\lceil \frac{\pi^2}{2\epsilon} Q_{\delta'}(f)\R\rceil
\end{equation}
Then the blowing up probability is no more than 
\begin{equation}
\frac{\pi^2}{2L} Q_{\delta'}(f) \le \epsilon
\end{equation}
and the difference between the final states, $\ket{\psi_f}$ and $\ket{\psi'_f}$, is at most
\begin{equation}
\L\| \ket{\psi'_{f}} - \ket{\psi_{f}} \R\| \le \frac{\pi^2}{4L} Q_{\delta'}(f) \le \frac{\epsilon}{2}.
\end{equation}

Therefore
\begin{align}
\L|\bra{\psi'_{f}}P\ket{\psi'_{f}} - \bra{\psi_{f}}P\ket{\psi_{f}}\R| &\le \L|\bra{\psi'_{f}}P\ket{\psi'_{f}} - \bra{\psi_{f}}P\ket{\psi'_{f}}\R| + \L|\bra{\psi'_{f}}P\ket{\psi_{f}} - \bra{\psi_{f}}P\ket{\psi_{f}}\R| \\
&\le \L\| \ket{\psi'_f} \R\| \L\| P\L(\ket{\psi'_f}-\ket{\psi_f}\R)\R\| + \L\| P\L(\ket{\psi'_f}-\ket{\psi_f}\R)\R\| \L\| \ket{\psi_f} \R\| \\
&\le \epsilon/2 + \epsilon/2 = \epsilon
\end{align}
for any projector $P$ (in particular, the projector that projects onto the classical answer at the end of the algorithm). The algorithm accumulates at most $\epsilon$ extra error at the end, giving a total error of no more than $\delta' + \epsilon = \delta$. This algorithm makes $2LQ_{\delta'}(f) < \frac{\pi^2}{\epsilon} Q_{\delta'}^2(f) + 2Q_{\delta'}(f)$ queries to the bomb oracle, and therefore
\begin{align}
B_{\epsilon,\delta} (f) &< \frac{\pi^2}{\epsilon}Q_{\delta-\epsilon}(f)^2 + 2Q_{\delta-\epsilon}(f)\\
&= O\L(\frac{Q_{\delta-\epsilon}(f)^2}{\epsilon}\R). \label{eq:temp upper bound}
\end{align}
From this we can derive that $B_{\epsilon,\delta} (f) = O(Q_{\delta}(f)^2/\epsilon)$:
\begin{align}
B_{\epsilon,\delta} (f) &< B_{\epsilon/2,\delta} (f) \\
&= O\L(\frac{Q_{\delta-\epsilon/2}(f)^2}{\epsilon}\R),\quad \text{by } \ref{eq:temp upper bound} \\
&= O\L(\frac{Q_{\delta}(f)^2}{\epsilon}\R),\quad \text{since } \frac{\delta}{2} \ge \delta - \frac{\epsilon}{2}.
\end{align}

\end{proof}

\subsection{Lower bound} \label{sect:lower bound}
\begin{thm} \label{thm:lowerbound}
For all functions $f$ with boolean input alphabet, and numbers $\epsilon$, $\delta$ satisfying $0 < \epsilon \le \delta \le 1/10$, 
\begin{equation}
B_{\epsilon,\delta}(f) = \Omega(Q_{0.01}(f)^2/\epsilon).
\end{equation}
\end{thm}

The proof of this result uses the generalized adversary bound $\text{Adv}^\pm(f)$ \cite{hoyer-2007}: we show that $B_{\epsilon}(f) = \Omega(\text{Adv}^\pm(f)^2/\epsilon)$, and then use the known result that $Q(f) = O(\text{Adv}^\pm(f))$ \cite{lee-2010}. The complete proof is given in Appendix \ref{app:lower bound proof}.

\section{Generalizations and Applications}
We now discuss applications of the result $B_{\epsilon}(f) = \Theta(Q(f)^2/\epsilon)$ that could be useful.
\subsection{Generalizing the bomb query model} \label{sect:generalize bomb}
We consider modifying the bomb query model as follows. We require that the input string $x$ can only be accessed by the following circuit:
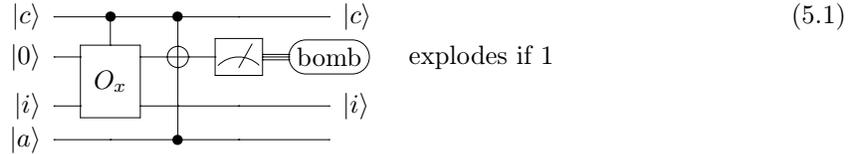
\begin{align}
\Qcircuit @C=1em @R=.7em {
\lstick{\ket{c}}& \ctrl{1} &\ctrl{1} &\qw & \rstick{\ket{c}} \qw \\
\lstick{\ket{0}}&   \multigate{1}{O_x} &\targ & \meter & \measure{\text{bomb}}\cw & \rstick{\text{explodes if 1}}     \\
\lstick{\ket{i}}&\ghost{O_x}  & \qw &\qw & \rstick{\ket{i}} \qw \\
\lstick{\ket{a}}&\qw  & \ctrl{-2} & \qw & \qw
}
\end{align}
Compare with Circuit \ref{circ:bomb}; the difference is that there is now an extra register $\ket{a}$, and the bomb explodes only if both $x_i=a$ and the control bit is 1. In other words, the bomb explodes if $c \cdot (x_i \oplus a) = 1$. The three registers $c$, $i$, and $a$ are allowed to be entangled, however. If we discard the second register afterwards, the effect of this circuit, written as a projector, is
\begin{align}
\tilde{M}_x = \sum_{i \in [N],a \in \{0,1\}}\ket{0,i,a}\bra{0,i,a} + \sum_{i,a: x_i = a}\ket{1,i,a}\bra{1,i,a}.
\end{align}
Let $\tilde{B}_{\epsilon,\delta}(f)$ be the required number of queries to this modified bomb oracle $\tilde{M}_x$ to calculate $f(x)$ with error no more than $\delta$, with a probability of explosion no more than $\epsilon$. Using Theorem \ref{thm:main result}, we show that $\tilde{B}$ and $B$ are equivalent up to a constant:
\begin{lem} \label{lem:equivalent}
If $f:\: D \rightarrow E$, where $D \subseteq \{0,1\}^N$, and $\delta \le 1/10$ is a constant, then $B_{\epsilon,\delta}(f) = \Theta(\tilde{B}_{\epsilon,\delta}(f))$.
\end{lem}
\begin{proof}
It should be immediately obvious that $B_{\epsilon,\delta}(f) \ge \tilde{B}_{\epsilon,\delta}(f)$, since a query in the $B$ model can be simulated by a query in the $\tilde{B}$ model by simply setting $a=0$. In the following we show that $B_{\epsilon,\delta}(f) = O(\tilde{B}_{\epsilon,\delta}(f))$.

For each string $x \in \{0,1\}^N$, define the string $\tilde{x} \in \{0,1\}^{2N}$ by concatenating two copies of $x$ and flipping every bit of the second copy. In other words,
\begin{align} \label{eq:tildex}
\tilde{x}_i = \begin{cases} x_i &\mbox{if } i \le N \\ 1-x_{i-N} &\mbox{if } i > N \end{cases}.
\end{align}
Let $\tilde{D} = \{\tilde{x}: x \in D\}$. Given a function $f:\: D \rightarrow \{0,1\}$, define $\tilde{f}:\: \tilde{D} \rightarrow \{0,1\}$ by $\tilde{f}(\tilde{x}) = f(x)$. 

We claim that a $\tilde{B}$ query to $x$ can be simulated by a $B$ query to $\tilde{x}$. This can be seen by comparing $\tilde{M}_x$:
\begin{align}
\tilde{M}_x = \sum_{i \in [N],a}\ket{0,i,a}\bra{0,i,a} + \sum_{i\in[N],a: x_i = a}\ket{1,i,a}\bra{1,i,a}.
\end{align}
and $M_{\tilde{x}}$:
\begin{align}
M_{\tilde{x}} = \sum_{\tilde{i} \in [2N]}\ket{0,\tilde{i}}\bra{0,\tilde{i}} + \sum_{\tilde{i} \in [2N]: \tilde{x}_i=0}\ket{1,\tilde{i}}\bra{1,\tilde{i}}.
\end{align}
Recalling the definition of $\tilde{x}$ in \ref{eq:tildex}, we see that these two projectors are exactly equal if we encode $\tilde{i}$ as $(i,a)$, where $i \equiv \tilde{i} \mod N$ and $a = \lfloor i / N \rfloor$.

Since $\tilde{f}(\tilde{x}) = f(x)$, we thus have $\tilde{B}_{\epsilon,\delta}(f) = B_{\epsilon,\delta}(\tilde{f})$. Our result then readily follows; it can easily be checked that $Q(f) = Q(\tilde{f})$, and therefore by Theorem \ref{thm:main result},
\begin{align}
\tilde{B}_{\epsilon,\delta}(f) &= B_{\epsilon,\delta}(\tilde{f}) = \Theta\L(\frac{Q(\tilde{f})^2}{\epsilon}\R) \nn \\
&= \Theta\L(\frac{Q(f)^2}{\epsilon}\R)
\end{align}
\end{proof}

There are some advantages to allowing the projector $\tilde{M}_x$ instead of $M_x$. First of all, the inputs 0 and 1 in $x$ are finally manifestly symmetric, unlike that in $M_x$ (the bomb originally blew up if $x_i=1$, but not if $x_i=0$). Moreover, we now allow the algorithm to \emph{guess} an answer to the query (this answer may be entangled with the index register $i$), and the bomb blows up only if the guess is wrong, controlled on $c$. This flexibility may allow more leeway in designing algorithms for the bomb query model, as we soon utilize.

\subsection{Using classical algorithms to design bomb query algorithms}
We now demonstrate the possibility that we can prove \emph{nonconstructive} upper bounds on $Q(f)$ for some functions $f$, by creating bomb query algorithms and using that $Q(f) = \Theta(\sqrt{\epsilon B_{\epsilon}(f)})$. Consider for example the following classical algorithm for the OR function:
\begin{alg}[Classical algorithm for OR]
Pick some arbitrary ordering of the $N$ bits, and query them one by one, terminating as soon as a 1 is seen. Return 1 if a 1 was queried; otherwise return 0.
\end{alg}
We can convert this immediately to a bomb query algorithm for OR, by using the construction in the proof of Theorem \ref{thm:upperbound}. That construction allows us to implement the operation $O_x$ in $O(\epsilon^{-1})$ queries, with $O(\epsilon)$ error and probability of explosion if $x_i=1$ (but no error if $x_i=0$). Thus we have the following:
\begin{alg}[Bomb algorithm for OR]
Query the $N$ bits one-by-one, and apply the construction of Theorem \ref{thm:upperbound} one bit at a time, using $O(1/\epsilon)$ operations each time. Terminate as soon as a 1 is seen, and return 1; otherwise return 0 if all bits are 0.
\end{alg}
Since the algorithm ends as soon as a 1 is found, the algorithm only accumulates $\epsilon$ error in total. Thus this shows $B_{\epsilon}(OR) = O(N/\epsilon)$.

Note, however, that we have already shown that $Q(f) = \Theta(\sqrt{\epsilon B_{\epsilon}(f)})$ for boolean $f$. An $O(N/\epsilon)$ bomb query algorithm for OR therefore implies that $Q(OR) = O(\sqrt{N})$. We have showed the existence of an $O(\sqrt{N})$ quantum algorithm for the OR function, without actually constructing one!

We formalize the intuition in the above argument by the following theorem: 
\begin{thm} \label{thm:classical}
Let $f:\:D \rightarrow E$, where $D \subseteq \{0,1\}^N$. Suppose there is a classical randomized query algorithm $\mathcal{A}$, that makes at most T queries, and evaluates $f$ with bounded error. Let the query results of $\mathcal{A}$ on random seed $s_{\mathcal{A}}$ be $x_{p_1},x_{p_2},\cdots,x_{p_{\tilde{T}(x)}}$, $\tilde{T}(x) \le T$, where $x$ is the hidden query string.

Suppose there is another (not necessarily time-efficient) randomized algorithm $\mathcal{G}$, with random seed $s_{\mathcal{G}}$, which takes as input $x_{p_1},\cdots,x_{p_{t-1}}$ and $s_\mathcal{A}$, and outputs a guess for the next query result $x_{p_{t}}$ of $\mathcal{A}$. Assume that $\mathcal{G}$ makes no more than an expected total of $G$ mistakes (for all inputs $x$). In other words,
\begin{equation}
\Expect_{s_{\mathcal{A}},s_{\mathcal{G}}} \L\{\sum_{t=1}^{\tilde{T}(x)}\L| \mathcal{G}(x_{p_1},\cdots,x_{p_{t-1}},s_{\mathcal{A}},s_{\mathcal{G}}) - x_{p_t} \R| \R\} \le G \quad \forall x.
\end{equation}
Note that $\mathcal{G}$ is given the random seed $s_{\mathcal{A}}$ of $\mathcal{A}$, so it can predict the next  query index of $\mathcal{A}$. \\
Then $B_{\epsilon}(f) = O(TG/\epsilon)$, and thus (by Theorem \ref{thm:main result}) $Q(f) = O(\sqrt{TG})$. 
\end{thm}

As an example, in our simple classical example for OR we have $T=N$ (the algorithm takes at most $N$ steps) and $G=1$ (the guessing algorithm always guesses the next query to be 0; since the algorithm terminates on a 1, it makes at most one mistake).

\begin{proof}[Proof of theorem \ref{thm:classical}]
We generalize the argument in the OR case. We take the classical algorithm and replace each classical query by the construction of Theorem \ref{thm:upperbound}, using $O(G/\epsilon)$ bomb queries each time. On each query, the bomb has a $O(\epsilon/G)$ chance of exploding when the guess is wrong, and no chance of exploding when the guess is correct. Therefore the total probability of explosion is $O(\epsilon/G)\cdot G = O(\epsilon)$. The total number of bomb queries used is $O(TG/\epsilon)$.

For the full technical proof, see Appendix \ref{app:nonconstructive}.
\end{proof}

\subsection{Explicit quantum algorithm for Theorem \ref{thm:classical}}
In this section we give an explicit quantum algorithm, in the setting of Theorem \ref{thm:classical}, that reproduces the given query complexity. This algorithm is very similar to the one given by R. Kothari for the oracle identification problem \cite{kothari-2014}.

\begin{thm} \label{thm:classical to quantum}
Under the assumptions of Theorem \ref{thm:classical}, there is an explicit quantum algorithm for $f$ with query complexity $O(\sqrt{TG})$.
\end{thm}

\begin{proof}
We will construct this algorithm (Algorithm \ref{alg:classical}) shortly. We need the following quantum search algorithm as a subroutine:
\begin{thm}[Finding the first marked element in a list] \label{thm:first marked element}
Suppose there is an ordered list of $N$ elements, and each element is either marked or unmarked. Then there is a bounded-error quantum algorithm for finding the \textbf{first} marked element in the list (or determines that no marked elements exist), such that:
\begin{itemize}
\item If the first marked element is the $d$-th element of the list, then the algorithm uses an expected $O(\sqrt{d})$ time and queries.
\item If there are no marked elements, then the algorithm uses $O(\sqrt{N})$ time and queries, but always determines correctly that no marked elements exist.
\end{itemize}
\end{thm}

This algorithm is straightforward to derive given the result in \cite{durr-1996}, and was already used in Kothari's algorithm \cite{kothari-2014}. We give the algorithm (Algorithm \ref{alg:first marked element}) and its analysis in Appendix \ref{app:first marked element}.

We now give our explicit quantum algorithm.

\begin{alg}[Simulating a classical query algorithm by a quantum one] \label{alg:classical}   \leavevmode

\setlength{\parskip}{0.5em}
\noindent \emph{Input.} Classical randomized algorithm $\mathcal{A}$ that computes $f$ with bounded error. Classical randomized algorithm $\mathcal{G}$ that guesses queries of $\mathcal{A}$. Oracle $O_x$ for the hidden string $x$.

\noindent \emph{Output.} $f(x)$ with bounded error.


The quantum algorithm proceeds by attempting to produce the list of queries and results that $\mathcal{A}$ would have made. More precisely, given a randomly chosen random seed $s_{\mathcal{A}}$, the quantum algorithm outputs (with constant error) a list of pairs $(p_1(x), x_{p_1(x)}),\cdots,(p_{\tilde{T}(x)}(x), x_{\tilde{T}(x)}{p_i(x)})$. This list is such that on random seed $s_\mathcal{A}$, the $i$-th query algorithm of $\mathcal{A}$ is made at the position $p_i(x)$, and the query result is $x_{p_i(x)}$. The quantum algorithm then determines the output of $\mathcal{A}$ using this list.

\setlength{\parskip}{0em}

The main idea for the algorithm is this: we first assume that the guesses made by $\mathcal{G}$ are correct. By repeatedly feeding the output of $\mathcal{G}$ back into $\mathcal{A}$ and $\mathcal{G}$, we can obtain a list of query values for $\mathcal{A}$ without any queries to the actual black box. We then search for the first deviation of the string $x$ from the predictions of $\mathcal{G}$; assuming the first deviation is the $d_1$-th query, by Theorem \ref{thm:first marked element} the search takes $O(\sqrt{d_1})$ queries (ignoring error for now). We then know that all the guesses made by $\mathcal{G}$ are correct up to the $(d_1-1)$-th query, and incorrect for the $d_1$-th query.

With the corrected result of the first $d_1$ queries, we now continue by assuming again the guesses made by $\mathcal{G}$ are correct starting from the $(d_1+1)$-th query, and search for the location of the next deviation, $d_2$. This takes $O(\sqrt{d_2-d_1})$ queries; we then know that all the guesses made by $\mathcal{G}$ are correct from the $(d_1+1)$-th to $(d_2-1)$-th query, and incorrect for the $d_2$-th one. Continuing in this manner, we eventually determine all query results of $\mathcal{A}$ after an expected $G$ iterations.

We proceed to spell out our algorithm. For the time being, we assume that algorithm for Theorem \ref{thm:first marked element} has no error and thus requires no error reduction.

\begin{enumerate}
\item Initialize random seeds $s_\mathcal{A}$ and $s_\mathcal{G}$ for $\mathcal{A}$ and $\mathcal{G}$. We will simulate the behavior of $\mathcal{A}$ and $\mathcal{G}$ on these random seeds.  Initialize $d= 0$. $d$ is such that we have determined the values of all query results of $\mathcal{A}$ up to the $d$-th query. Also initialize an empty list $\mathcal{L}$ of query pairs.
\item Repeat until either all query results of $\mathcal{A}$ are determined, or $100G$ iterations of this loop have been executed: \label{step:loop}
\begin{enumerate}
\item Assuming that $\mathcal{G}$ always guesses correctly starting from the $(d+1)$-th query, compute from $\mathcal{A}$ and $\mathcal{G}$ a list of query positions $p_{d+1},p_{d+2},\cdots$ and results $\tilde{a}_{d+1},\tilde{a}_{d+2},\cdots$. This requires no queries to the black box.
\item \label{step:d*}Using our algorithm for finding the first marked element (Theorem \ref{thm:first marked element}, Algorithm \ref{alg:first marked element}), find the first index $d^* > d$ such that the actual query result of $\mathcal{A}$ differs from the guess by $\mathcal{G}$, i.e. $x_{p_d} \neq \tilde{a}_d$; or return that no such $d^*$ exists. This takes $O(\sqrt{d^*-d })$ time in the former case, and $O(\sqrt{T-d})$ time in the latter.
\item We break into cases: 
\begin{enumerate}
\item If an index $d^*$ was found in Step \ref{step:d*}, then the algorithm decides the next mistake made by $\mathcal{G}$ is at position $d^*$. It thus adds the query pairs  $(p_{d+1},\tilde{a}_{d+1}),\cdots,(p_{d^*-1},\tilde{a}_{d^*-1})$, and the pair $(p_{d^*},1-\tilde{a}_{d^*})$, to the list $\mathcal{L}$. Also set $d=d^*$.
\item If no index $d^*$ was found in Step \ref{step:d*}, the algorithm decides that all remaining guesses by $\mathcal{G}$ are correct. Thus the query pairs  $(p_{d+1},\tilde{a}_{d+1}),\cdots,(p_{\tilde{T}(x)},\tilde{a}_{\tilde{T}(x)})$ are added to $\mathcal{L}$, where $\tilde{T}(x) \le T$ is the number of queries made by $\mathcal{A}$.
\end{enumerate}
\end{enumerate}
\item If the algorithm found all query results of $\mathcal{A}$ in $100G$ iterations of step  \ref{step:loop}, use $\mathcal{L}$ to calculate the output of $\mathcal{A}$; otherwise the algorithm fails.
\end{enumerate}

\end{alg}
We now count the total number of queries. Suppose $g \le 100G$ is the number of iterations of Step \ref{step:loop}; if all query results have been determined, $g$ is the number of wrong guesses by $\mathcal{G}$. Say the list of $d$'s found is $d_0 = 0, d_1,\cdots,d_{g}$. Let $d_{g+1} = T$. Step 2 is executed for $g+1$ times, and the total number of queries is
\begin{equation}
O\L(\sum_{i=1}^{g+1} \sqrt {d_i-d_{i-1}}\R) = O\L(\sqrt{Tg} \R) = O\L(\sqrt{TG} \R)
\end{equation}
by the Cauchy-Schwarz inequality.

We now analyze the error in our algorithm. The first source of error is cutting off the loop in Step \ref{step:loop}: by Markov's inequality, for at least 99\% of random seeds $s_{\mathcal{G}},s_{\mathcal{G}}$, $\mathcal{G}$ makes no more than $100G$ wrong guesses. For these random seeds all query results of $\mathcal{A}$ are determined. Cutting off the loop thus gives at most $0.01$ error. 

The other source of error is the error of Algorithm  \ref{alg:first marked element} used in Step \ref{step:d*}: we had assumed that it could be treated as zero-error, but we now remove this assumption. Assuming each iteration gives error $\delta'$, the total error accrued could be up to $O(g\delta')$. It seems as if we would need to set $\delta' = O(1/G)$ for the total error to be constant, and thus gain an extra logarithmic factor in the query complexity. 

However, in his paper for oracle identification \cite{kothari-2014}, Kothari showed that multiple calls to Algorithm \ref{alg:first marked element} can be composed to obtain a bounded-error algorithm based on span programs without an extra logarithmic factor in the query complexity; refer to \cite[Section 3]{kothari-2014} for details. Therefore we can replace the iterations of Step \ref{step:loop} with Kothari's span program construction and get a bounded error algorithm with complexity $O(\sqrt{TG})$.


\end{proof}

Note that while Algorithm \ref{alg:classical} has query complexity $O(\sqrt{TG})$, the time complexity may be much higher. After all, Algorithm \ref{alg:classical} proceeds by simulating $\mathcal{A}$ query-by-query, although the number of actual queries to the oracle is smaller. Whether or not we can get a algorithm faster than $\mathcal{A}$ using this approach may depend on the problem at hand.

\section{Improved upper bounds on quantum query complexity}
We now use Theorem \ref{thm:classical to quantum} to improve the quantum query complexity of certain graph problems.

\subsection{Single source shortest paths for unweighted graphs}
\begin{problem}[Single source shortest paths (SSSP) for unweighted graphs] The adjacency matrix of a directed graph $n$-vertex graph $G$ is provided as a black box; a query on the pair $(v,w)$ returns $1$ if there is an edge from $v$ to $w$, and $0$ otherwise. We are given a fixed vertex $v_{start}$. Call the length of a shortest path from $v_{start}$ to another vertex $w$ the \emph{distance} $d_w$ of $w$ from $v_{start}$; if no path exists, define $d_w = \infty$. Our task is to find $d_w$ for all vertices $w$ in $G$.
\end{problem}

In this section we shall show the following:

\begin{thm}
The quantum query complexity of single-source shortest paths in an unweighted graph is $\Theta(n^{3/2})$ in the adjacency matrix model.
\end{thm}

\begin{proof}
The lower bound of $\Omega(n^{3/2})$ is known \cite{durr-2004}. We show the upper bound by applying Theorem \ref{thm:classical to quantum} to a classical algorithm. The following well-known classical algorithm (commonly known as \emph{breadth first search}, BFS) solves this problem:
\begin{alg}[Classical algorithm for unweighted SSSP] \label{alg:bfs} \leavevmode
\begin{enumerate}
\item Initialize $d_w := \infty$ for all vertices $w\neq v_{start}$, $d_{v_{start}}:=0$, and $\mathcal{L}:=(v_{start})$. $\mathcal{L}$ is the ordered list of vertices for which we have determined the distances, but whose outgoing edges we have not queried.
\item Repeat until $\mathcal{L}$ is empty:
\begin{itemize}
\item Let $v$ be the first (in order of time added to $\mathcal{L}$) vertex in $\mathcal{L}$. For all vertices $w$ such that $d_w = \infty$:
\begin{itemize}
\item Query $(v,w)$.
\item If $(v,w)$ is an edge, set $d_w:= d_v + 1$  and add $w$ to the end of $\mathcal{L}$.
\end{itemize}
\item Remove $v$ from $\mathcal{L}$.
\end{itemize}
\end{enumerate}
\end{alg}
We omit the proof of correctness of this algorithm (see for example \cite{clrs}). This algorithm uses up to $T=O(n^2)$ queries. If the guessing algorithm always guesses that $(v,w)$ is not an edge, then it makes at most $G=n-1$ mistakes; hence $Q(f) = O(\sqrt{TG}) = O(n^{3/2})$.\footnote{It seems difficult to use our method to give a corresponding result for the adjacency list model; after all, the result of a query is much harder to guess when the input alphabet is non-boolean.}

\end{proof}

The previous best known quantum algorithm for unweighted SSSP, to our best knowledge, was given by Furrow \cite{furrow-2008}; that algorithm has query complexity $O(n^{3/2}\sqrt{\log n})$.

We now consider the quantum query complexity of unweighted $k$-source shortest paths (finding $k$ shortest-path trees rooted from $k$ beginning vertices). If we apply Algorithm \ref{alg:bfs} on $k$ different starting vertices, then the expected number of wrong guesses is no more than $G = k(n-1)$; however, the total number of edges we query need not exceed $T = O(n^2)$, since an edge never needs to be queried more than once. Therefore
\begin{cor}
The quantum query complexity of unweighted $k$-source shortest paths in the adjacency matrix model is $O(k^{1/2}n^{3/2})$, where $n$ is the number of vertices.
\end{cor}
We use this idea -- that $T$ need not exceed $O(n^2)$ when dealing with graph problems -- again in the following section.

\subsection{Maximum bipartite matching}
\begin{problem}[Maximum bipartite matching] We are given as black box the adjacency matrix of an $n$-vertex bipartite graph $G=(V=X\cup Y,E)$, where the undirected set of edges $E$ only run between the bipartite components $X$ and $Y$. A \emph{matching} of $G$ is a list of edges of $G$ that do not share vertices. Our task is to find a maximum matching of $G$, i.e. a matching that contains the largest possible number of edges.
\end{problem}

In this section we show that
\begin{thm}
The quantum query complexity of maximum bipartite matching is $O(n^{7/4})$ in the adjacency matrix model, where $n$ is the number of vertices.
\end{thm}

\begin{proof}

Once again we apply Theorem \ref{thm:classical to quantum} to a classical algorithm. Classically, this problem is solved in $O(n^{5/2})$ time by the Hopcroft-Karp \cite{hopcroft-1973} algorithm (here $n = |V|$). We summarize the algorithm as follows (this summary roughly follows that of \cite{ambainis-2006}):
\begin{alg}[Hopcroft-Karp algorithm for maximum bipartite matching {\cite{hopcroft-1973}}] \label{alg:hopcroft-karp} \leavevmode

\begin{enumerate}
\item Initialize an empty matching $\mathcal{M}$. $\mathcal{M}$ is a matching that will be updated until it is maximum.
\item Repeat the following steps until $\mathcal{M}$ is a maximum matching:
\begin{enumerate}
\item Define the \emph{directed} graph $H=(V',E')$ as follows:
\begin{align}
V' &= X \cup Y \cup \{s,t\}  \nonumber \\
E' &= \{(s,x)\: | \: x \in X, (x,y) \not\in \mathcal{M} \: \text{ for all } y \in Y\} \nonumber \\
&\cup \{(x,y)\: | \: x \in X, y \in Y, (x,y) \in E, (x,y)\not\in\mathcal{M}\} \nonumber \\
&\cup \{(y,x)\: | \: x \in X, y \in Y, (x,y) \in E, (x,y)\in\mathcal{M}\} \nonumber \\
&\cup \{(y,t)\: | \: y \in Y, (x,y) \not\in \mathcal{M} \: \text{ for all } x \in X\}
\end{align}
where $s$ and $t$ are two extra auxilliary vertices. Note that if $(s,x_1,y_1,x_2,y_2,\cdots,x_\ell,y_\ell,t)$ is a path in $H$ from $s$ to $t$, then $x_i \in X$ and $y_i \in Y$ for all $i$. Additionally, the edges (aside from the first and last) alternate from being in $\mathcal{M}$ and not being in $\mathcal{M}$: $(x_i,y_i) \not\in\mathcal{M}$, $(y_i,x_{i+1}) \in \mathcal{M}$. Such a path is called an \emph{augmenting path} in the literature.

We note that a query to the adjacency matrix of $E'$ can be simulated by a query to the adjacency matrix of $E$.

\item Using the breadth-first search algorithm (Algorithm \ref{alg:bfs}), in the graph $H$, find the length of the shortest path, or distance, of all vertices from $s$. Let the distance from $s$ to $t$ be $2\ell+1$.
\item Find a maximal set $S$ of vertex-disjoint shortest paths from $s$ to $t$ in the graph $H$. In other words, $S$ should be a list of paths from $s$ to $t$ such that each path has length $2\ell+1$, and no pair of paths share vertices except for $s$ and $t$. Moreover, all other shortest paths from $s$ to $t$ share at least one vertex (except for $s$ and $t$) with a path in $S$. We describe how to find such a maximal set in Algorithm \ref{alg:modified dfs}.
\item If $S$ is empty, the matching $M$ is a maximum matching, and we terminate. Otherwise continue:
\item Let $(s,x_1,y_1,x_2,y_2,\cdots,x_\ell,y_\ell,t)$ be a path in $S$. Remove the $\ell-1$ edges $(x_{i+1},y_i)$ from $\mathcal{M}$, and insert the $\ell$ edges $(x_i,y_i)$ into $\mathcal{M}$. This increases $|\mathcal{M}|$ by 1. Repeat for all paths in $S$; there are no conflicts since the paths in $S$ are vertex-disjoint.
\end{enumerate}
\end{enumerate}
\end{alg}

Once again, we omit the proof of correctness of this algorithm; the correctness is guaranteed by Berge's Lemma \cite{berge-1957}, which states that a matching is maximum if there are no more augmenting paths for the matching. Moreover, $O(\sqrt{n})$ iterations of Step 2 suffice \cite{hopcroft-1973}.

We now describe how to find a maximal set of shortest-length augmenting paths in Step 2(c). This algorithm is essentially a modified version of depth-first search:
\begin{alg}[Finding a maximal set of vertex-disjoint shortest-length augmenting paths] \label{alg:modified dfs} \leavevmode
\begin{alg_in}
The directed graph $H$ defined in Algorithm \ref{alg:hopcroft-karp}, as well as the distances $d_v$ of all vertices $v$ from $s$ (calculated in Step 2(b) of Algorithm \ref{alg:hopcroft-karp}).
\end{alg_in}
\begin{enumerate}
\item Initialize a set of paths $S := \emptyset$, set of vertices $R:=\{s\}$, and a stack\footnote{A stack is a data structure such that elements that are first inserted into the stack are removed last.} of vertices $\mathcal{L}:=(s)$. $\mathcal{L}$ contains the ordered list of vertices that we have begun, but not yet finished, processing. $R$ is the set of vertices that we have processed. $S$ is the set of vertex-disjoint shortest-length augmenting paths that we have found.
\item Repeat until $\mathcal{L}$ is empty:
\begin{enumerate}
\item If the vertex in the front of $\mathcal{L}$ is $t$, we have found a new vertex-disjoint path from $s$ to $t$:
\begin{itemize}
\item Trace the path from $t$ back to $s$ by removing elements from the front of $\mathcal{L}$ until $s$ is at the front. Add the corresponding path to $S$.
\item Start again from the beginning of Step 2.
\end{itemize}
\item Let $v$ be the vertex in the front of $\mathcal{L}$ (i.e. the vertex \emph{last} added to, and still in, $\mathcal{L}$). Recall the distance from $s$ to $v$ is $d_v$.
\item Find $w$ such that $w \not\in R$, $d_w = d_v+1$, and $(v,w)$ (as an edge in $H$) has not been queried in this algorithm. If no such vertex $w$ exists, remove $v$ from $\mathcal{L}$ and start from the beginning of Step 2.
\item Query $(v,w)$ on the graph $H$.
\item If $(v,w)$ is an edge, add $w$ to the \emph{front} of $\mathcal{L}$. If $w \neq t$, add $w$ to $R$.
\end{enumerate}
\item Output $S$, the maximal set of vertex-disjoint shortest-length augmenting paths.
\end{enumerate}
\end{alg}

We now return to Algorithm \ref{alg:hopcroft-karp} and count $T$ and $G$. There is obviously no need to query the same edge more than once, so $T = O(n^2)$. If the algorithm always guesses, on a query $(v,w)$, that there is no edge between $(v,w)$, then it makes at most $G = O(n^{3/2})$ mistakes: in Step 2(b) there are at most $O(n)$ mistakes (see Algorithm \ref{alg:bfs}), while in Step 2(c)/Algorithm \ref{alg:modified dfs} there is at most one queried edge leading to each vertex aside from $t$, and edges leading to $t$ can be computed without queries to $G$. Since Step 2 is executed $O(\sqrt{n})$ times, our counting follows.

Thus there is a quantum query algorithm with complexity $Q= O(\sqrt{TG}) = O(n^{7/4})$.

\end{proof}

To our knowledge, this is the first known nontrivial upper bound on the query complexity of maximum bipartite matching.\footnote{The trivial upper bound is $O(n^2)$, where all pairs of vertices are queried.} The time complexity of this problem was studied by Ambainis and Spalek in \cite{ambainis-2006}; they gave an upper bound of $O(n^2\log n)$ time in the adjacency matrix model. A lower bound of $\Omega(n^{3/2})$ for the query complexity of this problem was given in \cite{berzina-2004,zhang-2005}.

For readers familiar with network flow, the arguments in this section also apply to Dinic's algorithm for maximum flow \cite{dinic-1970} on graphs with unit capacity, i.e. where the capacity of each edge is 0 or 1. On graphs with unit capacity, Dinic's algorithm is essentially the same as Hopcroft-Karp's, except that augmenting paths are over a general, nonbipartite flow network. (The set $S$ in Step 2(c) of Algorithm \ref{alg:hopcroft-karp} is generally referred to as a \emph{blocking flow} in this context.) It can be shown that only $O(\min\{m^{1/2},n^{2/3}\})$ iterations of Step 2 are required \cite{karzanov-1973,even-1975}, where $m$ is the number of edges of the graph. Thus $T=O(n^2)$, $G = O(\min\{m^{1/2},n^{2/3}\}n)$, and therefore

\begin{thm}
The quantum query complexity of the maximum flow problem in graphs with unit capacity is $O(\min\{n^{3/2}m^{1/4},n^{11/6}\})$, where $n$ and $m$ are the number of vertices and edges in the graph, respectively.
\end{thm}

It is an open question whether a similar result for maximum matching in a general nonbipartite graph can be proven, perhaps by applying Theorem \ref{thm:classical to quantum} to the classical algorithm of Micali and Vazirani \cite{micali-1980}.

\section{Projective query complexity} \label{sect:projective query}
We end this paper with a brief discussion on another query complexity model, which we will call the \emph{projective query complexity}.  This model is similar to the bomb query model in that the only way of accessing $x_i$ is through a classical measurement; however, in the projective query model the algorithm does not terminate if a 1 is measured. Our motivation for considering the projective query model is that its power is intermediate between the classical and quantum query models. To the best of our knowledge, this model was first considered in 2002 in unpublished results by S. Aaronson \cite{aaronson-2014}.

A circuit in the projective query complexity model is a restricted quantum query circuit, with the following restrictions on the use of the quantum oracle:

\begin{enumerate}
\item We have an extra control register $\ket{c}$ used to control whether $O_x$ is applied (we call the controlled version $CO_x$):
\begin{equation}
CO_x\ket{c,r,i} = \ket{c,r\oplus(c\cdot x_i),i}.
\end{equation}
where $\cdot$ indicates boolean AND.
\item The record register, $\ket{r}$ in the definition of $CO_x$ above, \emph{must} contain $\ket{0}$ before $CO_x$ is applied.
\item After $CO_x$ is applied, the record register is immediately measured in the computational basis, giving the answer $c\cdot x_i$. The result, a classical bit, can then be used to control further quantum unitaries (although only controlling the next unitary is enough, since the classical bit can be stored).

\end{enumerate}

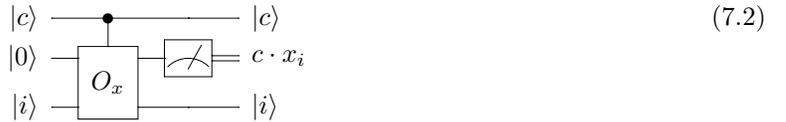
\begin{align} \label{circ:cp}
&\Qcircuit @C=1em @R=.7em {
\lstick{\ket{c}}& \ctrl{1} &\qw & \rstick{\ket{c}} \qw \\
\lstick{\ket{0}}&   \multigate{1}{O_x} &\meter & \rstick{c \cdot x_i} \cw  \\
\lstick{\ket{i}}&\ghost{O_x}  &\qw & \rstick{\ket{i}} \qw
}
\end{align}

We wish to evaluate a function $f(x)$ with as few calls to this \emph{projective oracle} as possible. Let the number of oracle calls required to evaluate $f(x)$, with at most $\delta$ error, be $P_{\delta}(f)$. By gap amplification, the choice of $\delta$ only affects $P_{\delta}(f)$ by a factor of $\log(1/\delta)$, and thus we will often omit $\delta$.

We can compare the definition in this section with the definition of the bomb query complexity in Section \ref{sect:bomb query}: the only difference is that if $c \cdot x_i=1$, the algorithm terminates in the bomb model, while the algorithm can continue in the projective model. Therefore the following is evident:

\begin{obs}
$P_{\delta}(f) \le B_{\epsilon,\delta}(f)$, and therefore $P(f) = O(Q(f)^2)$.
\end{obs}

Moreover, it is clear that the projective query model has power intermediate between classical and quantum (a controlled query in the usual quantum query model can be simulated by appending a 0 to the input string), and therefore letting $R_{\delta}(f)$ be the classical randomized query complexity,

\begin{obs}
$Q_{\delta}(f) \le P_{\delta}(f) \le R_{\delta}(f)$.
\end{obs}

For explicit bounds on $P$, Regev and Schiff \cite{regev-2008} have shown that for computing the OR function, the projective query complexity loses the Grover speedup:
\begin{thm}[{\cite{regev-2008}}] \label{thm:por}
$P(OR) = \Omega(N)$.
\end{thm}
Note that this result says nothing about $P(AND)$, since the definition of $P(f)$ is asymmetric with respect to 0 and 1 in the input.\footnote{We could have defined a symmetric version of $P$, say $\tilde{P}$, by allowing an extra guess on the measurement result, similar to our construction of $\tilde{B}$ in Section \ref{sect:generalize bomb}. Unfortunately, Regev and Schiff's result, Theorem \ref{thm:por}, do not apply to this case, and we see no obvious equivalence between $P$ and $\tilde{P}$.}

We observe that there could be a separation in both parts of the inequality $Q \le P \le B$:
\begin{equation}
Q(OR) = \Theta(\sqrt{N}),\quad P(OR) = \Theta(N),\quad B(OR) = \Theta(N)
\end{equation}
\begin{equation}
Q(PARITY) = \Theta(N),\quad P(PARITY) = \Theta(N),\quad B(PARITY) = \Theta(N^2)
\end{equation}
In the latter equation we used the fact that $Q(PARITY)=\Theta(N)$ \cite{beals-1998}. It therefore seems difficult to adapt our lower bound method in Section \ref{sect:lower bound} to $P(f)$.

It would be interesting to find a general lower bound for $P(f)$, or to establish more clearly the relationship between $Q(f)$, $P(f)$, and $R(f)$.

\section*{Acknowledgements}
We are grateful to Scott Aaronson and Aram Harrow for many useful discussions, and Scott Aaronson and Shelby Kimmel for valuable suggestions on a preliminary draft. We also thank Andrew Childs for giving us permission to make use of his online proof of the general adversary lower bound in \cite{childs-2013}. Special thanks to Robin Kothari for pointing us to his paper \cite{kothari-2014}, and in particular his analysis showing that logarithmic factors can be removed from the query complexity of Algorithm \ref{alg:classical}. This work is supported by the ARO grant Contract Number W911NF-12-0486. CYL gratefully acknowledges support from the Natural Sciences and Engineering Research Council of Canada. 

\bibliographystyle{utphys}
\bibliography{bomb_query}

\providecommand{\href}[2]{#2}\begingroup\raggedright\begin{thebibliography}{10}

\bibitem{elitzur-1993}
A.~C. Elitzur and L.~Vaidman, ``Quantum mechanical interaction-free
  measurements,'' {\em Foundations of Physics} {\bfseries 23} no.~7, (July,
  1993) 987--997, \href{http://arxiv.org/abs/hep-th/9305002}{{
  arXiv:hep-th/9305002}}.

\bibitem{furrow-2008}
B.~Furrow, ``A panoply of quantum algorithms,'' {\em Quantum Information and
  Computation} {\bfseries 8} no.~8, (September, 2008) 834--859,
  \href{http://arxiv.org/abs/quant-ph/0606127}{{ arXiv:quant-ph/0606127}}.

\bibitem{grover-1997}
L.~K. Grover, ``A fast quantum mechanical algorithm for database search,'' in
  {\em Proceedings of the 28th Annual ACM Symposium on the Theory of Computing
  (STOC)}.
\newblock May, 1996.
\newblock \href{http://arxiv.org/abs/quant-ph/9605043}{{
  arXiv:quant-ph/9605043}}.

\bibitem{bbbv-1997}
C.~H. Bennett, E.~Bernstein, G.~Brassard, and U.~Vazirani, ``Strengths and
  weaknesses of quantum computing,'' {\em SIAM Journal on Computing} {\bfseries
  26} no.~5, (1997) 1510--1523, \href{http://arxiv.org/abs/quant-ph/9701001}{{
  arXiv:quant-ph/9701001}}.

\bibitem{ambainis-2003}
A.~Ambainis, ``Quantum walk algorithm for element distinctness,'' {\em SIAM
  Journal on Computing} {\bfseries 37} no.~1, (2007) 210--239,
  \href{http://arxiv.org/abs/quant-ph/0311001}{{ arXiv:quant-ph/0311001}}.

\bibitem{szegedy-2004}
M.~Szegedy, ``Quantum speed-up of {M}arkov chain based algorithms,'' in {\em
  Proceedings of the 45th Annual IEEE Symposium on Foundations of Computer
  Science (FOCS)}.
\newblock 2004.

\bibitem{magniez-2006}
F.~Magniez, A.~Nayak, J.~Roland, and M.~Santha, ``Search via quantum walk,''
  {\em SIAM Journal on Computing} {\bfseries 40} no.~1, (2011) 142--164,
  \href{http://arxiv.org/abs/quant-ph/0608026}{{ arXiv:quant-ph/0608026}}.

\bibitem{jeffery-2012}
S.~Jeffery, R.~Kothari, and F.~Magniez, ``Nested quantum walks with quantum
  data structures,'' \href{http://arxiv.org/abs/1210.1199}{{ arXiv:1210.1199
  [quant-ph]}}.

\bibitem{belovs-2011}
A.~Belovs, ``Span programs for functions with constant-sized 1-certificates,''
  \href{http://arxiv.org/abs/1105.4024}{{ arXiv:1105.4024 [quant-ph]}}.

\bibitem{beals-1998}
R.~Beals, H.~Buhrman, R.~Cleve, M.~Mosca, and R.~de~Wolf, ``Quantum lower
  bounds by polynomials,'' in {\em Proceedings of the 39th Annual Symposium on
  Foundations of Computer Science (FOCS)}, p.~352.
\newblock 1998.
\newblock \href{http://arxiv.org/abs/quant-ph/9802049}{{
  arXiv:quant-ph/9802049}}.

\bibitem{ambainis-2000}
A.~Ambainis, ``Quantum lower bounds by quantum arguments,'' in {\em Proceedings
  of the 32nd Annual ACM Symposium on Theory of Computing (STOC)},
  pp.~636--643.
\newblock 2000.
\newblock \href{http://arxiv.org/abs/quant-ph/0002066}{{
  arXiv:quant-ph/0002066}}.

\bibitem{hoyer-2007}
P.~H{\o}yer, T.~Lee, and R.~\v{S}palek, ``Negative weights make adversaries
  stronger,'' in {\em Proceedings of the 39th Annual ACM Symposium on Theory of
  Computing (STOC)}, pp.~526--535.
\newblock 2007.
\newblock \href{http://arxiv.org/abs/quant-ph/0611054}{{
  arXiv:quant-ph/0611054}}.

\bibitem{reichardt-2009}
B.~W. Reichardt, ``Span programs and quantum query complexity: The general
  adversary bound is nearly tight for every boolean function,'' in {\em
  Proceedings of the 50th IEEE Symposium on Foundations of Computer Science
  (FOCS)}, pp.~544--551.
\newblock 2009.
\newblock \href{http://arxiv.org/abs/0904.2759}{{ arXiv:0904.2759 [quant-ph]}}.

\bibitem{reichardt-2011}
B.~W. Reichardt, ``Reflections for quantum query algorithms,'' in {\em
  Proceedings of the 22nd ACM-SIAM Symposium on Discrete Algorithms (SODA)},
  pp.~560--569.
\newblock 2011.
\newblock \href{http://arxiv.org/abs/1005.1601}{{ arXiv:1005.1601 [quant-ph]}}.

\bibitem{lee-2010}
T.~Lee, R.~Mittal, B.~W. Reichardt, R.~\v{S}palek, and M.~Szegedy, ``Quantum
  query complexity of state conversion,'' in {\em Proceedings of the 52nd IEEE
  Symposium on Foundations of Computer Science (FOCS)}, pp.~344--353.
\newblock 2011.
\newblock \href{http://arxiv.org/abs/1011.3020}{{ arXiv:1011.3020 [quant-ph]}}.

\bibitem{kimmel-2013}
S.~Kimmel, ``Quantum adversary (upper) bound,'' {\em Chicago Journal of
  Theoretical Computer Science} no.~4, (2013) ,
  \href{http://arxiv.org/abs/1101.0797}{{ arXiv:1101.0797 [quant-ph]}}.

\bibitem{misra-1977}
B.~Misra and E.~C.~G. Sudarshan, ``The {Z}eno's paradox in quantum theory,''
  {\em Journal of Mathematical Physics} {\bfseries 18} no.~4, (1977) 756.

\bibitem{kwiat-1995}
P.~Kwiat, H.~Weinfurter, T.~Herzog, A.~Zeilinger, and M.~A. Kasevich,
  ``Interaction-free measurement,'' {\em Physical Review Letters} {\bfseries
  74} no.~24, (1995) 4763.

\bibitem{ambainis-2006}
A.~Ambainis and R.~\v{S}palek, ``Quantum algorithms for matching and network
  flows,'' in {\em Lecture Notes in Computer Science}, vol.~3884, pp.~172--183.
\newblock Springer, 2006.
\newblock \href{http://arxiv.org/abs/quant-ph/0508205}{{
  arXiv:quant-ph/0508205}}.

\bibitem{dorn-2008}
S.~D\"{o}rn, ``Quantum algorithms for matching problems,'' {\em Theory of
  Computing Systems} {\bfseries 45} no.~3, (October, 2009) 613--628.

\bibitem{regev-2008}
O.~Regev and L.~Schiff, ``Impossibility of a quantum speed-up with a faulty
  oracle,'' in {\em Lecture Notes in Computer Science}, vol.~5125,
  pp.~773--781.
\newblock Springer, 2008.
\newblock \href{http://arxiv.org/abs/1202.1027}{{ arXiv:1202.1027 [quant-ph]}}.

\bibitem{mitchison-2001}
G.~Mitchison and R.~Jozsa, ``Counterfactual computation,'' {\em Proceedings of
  the Royal Society A} {\bfseries 457} no.~2009, (2001) 1175--1194,
  \href{http://arxiv.org/abs/quant-ph/9907007}{{ arXiv:quant-ph/9907007}}.

\bibitem{hosten-2006}
O.~Hosten, M.~T. Rakher, J.~T. Barreiro, N.~A. Peters, and P.~G. Kwiat,
  ``Counterfactual quantum computation through quantum interrogation,'' {\em
  Nature} {\bfseries 439} (February, 2006) 949--952.

\bibitem{mitchison-2006}
G.~Mitchison and R.~Jozsa, ``The limits of counterfactual computation,''
  \href{http://arxiv.org/abs/quant-ph/0606092}{{ arXiv:quant-ph/0606092}}.

\bibitem{hosten-2006-2}
O.~Hosten, M.~T. Rakher, J.~T. Barreiro, N.~A. Peters, and P.~Kwiat,
  ``Counterfactual computation revisited,''
  \href{http://arxiv.org/abs/quant-ph/0607101}{{ arXiv:quant-ph/0607101}}.

\bibitem{vaidman-2006}
L.~Vaidman, ``The impossibility of the counterfactual computation for all
  possible outcomes,'' \href{http://arxiv.org/abs/quant-ph/0610174}{{
  arXiv:quant-ph/0610174}}.

\bibitem{hosten-2006-3}
O.~Hosten and P.~G. Kwiat, ``Weak measurements and counterfactual
  computation,'' \href{http://arxiv.org/abs/quant-ph/0612159}{{
  arXiv:quant-ph/0612159}}.

\bibitem{salih-2013}
H.~Salih, Z.-H. Li, M.~Al-Amri, and M.~S. Zubairy, ``Protocol for direct
  counterfactual quantum communication,'' {\em Physical Review Letters}
  {\bfseries 110} (2013) 170502, \href{http://arxiv.org/abs/1206.2042}{{
  arXiv:1206.2042 [quant-ph]}}.

\bibitem{vaidman-2013}
L.~Vaidman, ``Comment on "protocol for direct counterfactual quantum
  communication" [arxiv:1206.2042],'' \href{http://arxiv.org/abs/1304.6689}{{
  arXiv:1304.6689 [quant-ph]}}.

\bibitem{noh-2009}
T.-G. Noh, ``Counterfactual quantum cryptography,'' {\em Physical Review
  Letters} {\bfseries 103} (2009) 230501,
  \href{http://arxiv.org/abs/0809.3979}{{ arXiv:0809.3979 [quant-ph]}}.

\bibitem{brodutch-2014}
A.~Brodutch, D.~Nagaj, O.~Sattath, and D.~Unruh, ``An adaptive attack on
  {W}iesner's quantum money,'' \href{http://arxiv.org/abs/1404.1507}{{
  arXiv:1404.1507 [quant-ph]}}.

\bibitem{wiesner-1983}
S.~Wiesner, ``Conjugate coding,'' {\em ACM SIGACT News} {\bfseries 15} no.~1,
  (1983) .

\bibitem{kothari-2014}
R.~Kothari, ``An optimal quantum algorithm for the oracle identification
  problem,'' in {\em Proceedings of the 31st International Symposium on
  Theoretical Aspects of Computer Science (STACS)}, E.~W. Mayr and N.~Portier,
  eds., vol.~25 of {\em Leibniz International Proceedings in Informatics
  (LIPIcs)}, pp.~482--493.
\newblock Schloss Dagstuhl--Leibniz-Zentrum fuer Informatik, Dagstuhl, Germany,
  2014.
\newblock \href{http://arxiv.org/abs/1311.7685}{{ arXiv:1311.7685 [quant-ph]}}.

\bibitem{aaronson-2014}
S.~Aaronson. Personal communication, 2014.

\bibitem{micali-1980}
S.~Micali and V.~V. Vazirani, ``An ${O}(\sqrt{|V|}\cdot |{E}|)$ algorithm for
  finding maximum matching in general graphs,'' in {\em Proceedings of the 21st
  Annual Symposium on Foundations of Computer Science (FOCS)}, pp.~17--27.
\newblock 1980.

\bibitem{cook-1986}
S.~Cook, C.~Dwork, and R.~Reischuk, ``Upper and lower time bounds for parallel
  random access machines without simultaneous writes,'' {\em SIAM Journal on
  Computing} {\bfseries 15} no.~1, (1986) 87--97.

\bibitem{nisan-1991}
N.~Nisan, ``{CREW PRAM}s and decision trees,'' {\em SIAM Journal on Computing}
  {\bfseries 20} no.~6, (1991) 999--1007.

\bibitem{buhrman-1999}
H.~Buhrman and R.~D. Wolf, ``Complexity measures and decision tree complexity:
  A survey,'' {\em Theoretical Computer Science} {\bfseries 288} (1999) 2002.

\bibitem{durr-1996}
C.~D\"{u}rr and P.~H{\o}yer, ``A quantum algorithm for finding the minimum,''
  \href{http://arxiv.org/abs/quant-ph/9607014}{{ arXiv:quant-ph/9607014}}.

\bibitem{durr-2004}
C.~D\"{u}rr, M.~Heiligman, P.~H{\o}yer, and M.~Mhalla, ``Quantum query
  complexity of some graph problems,''
  \href{http://arxiv.org/abs/quant-ph/0401091}{{ arXiv:quant-ph/0401091}}.

\bibitem{clrs}
T.~H. Cormen, C.~E. Leiserson, R.~Rivest, and C.~Stein, {\em Introduction to
  Algorithms}.
\newblock MIT Press and McGraw-Hill, 3rd~ed., 2009.

\bibitem{hopcroft-1973}
J.~E. Hopcroft and R.~M. Karp, ``An $n^{5/2}$ algorithm for maximum matchings
  in bipartite graphs,'' {\em SIAM Journal on Computing} {\bfseries 2} no.~4,
  (1973) 225--231.

\bibitem{berge-1957}
C.~Berge, ``Two theorems in graph theory,'' in {\em Proceedings of the National
  Academy of Sciences of the United States of America}, vol.~43, pp.~842--844.
\newblock 1957.

\bibitem{berzina-2004}
A.~Berzina, A.~Dubrovsky, R.~Freivalds, L.~Lace, and O.~Scegulnaja, ``Quantum
  query complexity for some graph problems,'' in {\em Lecture Notes in Computer
  Science}, vol.~2932, pp.~140--150.
\newblock Springer, 2004.

\bibitem{zhang-2005}
S.~Zhang, ``On the power of {A}mbainis's lower bounds,'' {\em Theoretical
  Computer Science} {\bfseries 339} no.~2-3, (2005) 241--256,
  \href{http://arxiv.org/abs/quant-ph/0311060}{{ arXiv:quant-ph/0311060}}.

\bibitem{dinic-1970}
E.~A. Dinic, ``Algorithm for solution of a problem of maximum flow in a network
  with power estimation,'' {\em Soviet Math Doklady} {\bfseries 11} (1970)
  1277--1280.

\bibitem{karzanov-1973}
A.~V. Karzanov, ``O nakhozhdenii maksimal'nogo potoka v setyakh spetsial'nogo
  vida i nekotorykh prilozheniyakh,'' in {\em Matematicheskie Voprosy
  Upravleniya Proizvodstvom}, L.~Lyusternik, ed., vol.~5, pp.~81--94.
\newblock Moscow State University Press, 1973.

\bibitem{even-1975}
S.~Even and R.~E. Tarjan, ``Network flow and testing graph connectivity,'' {\em
  SIAM Journal on Computing} {\bfseries 4} no.~4, (1975) 507--518.

\bibitem{childs-2013}
A.~Childs. \url{http://www.math.uwaterloo.ca/~amchilds/teaching/w13/l15.pdf},
  2013.

\bibitem{bhatia-1997}
R.~Bhatia, {\em Matrix Analysis}.
\newblock Springer-Verlag, 1997.

\end{thebibliography}\endgroup

\appendix

\section{Proof of the adversary lower bound for $B(f)$ (Theorem \ref{thm:lowerbound})} \label{app:lower bound proof}
Before we give the proof of the general result that $B(f) = \Omega(Q(f)^2)$ (Theorem \ref{thm:lowerbound})), we will illustrate the proof by means of an example, the special case where $f$ is the AND function.
\begin{thm}
For $\delta < 1/10$, $B_{\epsilon,\delta}(AND)=\Omega({N \over \epsilon})$.
\end{thm}

\begin{proof}
Let $\ket{\psi^0_t}$ be the unnormalized state of the algorithm with $x=1^n$, and $\ket{\psi^k_t}$ be the unnormalized state with $x=1\cdots 101\cdots 1$, $x_k=0$, right before the $(t+1)$-th call to $M_{x}$. Then
\begin{align}
\ket{\psi^x_{t+1}} = U_{t+1}M_x\ket{\psi^x_t}
\end{align}
for some unitary $U_{t+1}$. For ease of notation, we'll write $M_0 \equiv M_{1^n}$ and $M_k = M_{1\cdots 101\cdots 1}$, where the $k$-th bit is 0 in the latter case. When acting on the control and index bits,
\begin{align}
M_0 &= \sum\limits_{i=1}^{N}\ket{0,i}\bra{0,i} \nn \\
M_k &= \sum\limits_{i=1}^N\ket{0,i}\bra{0,i} + \ket{1,k}\bra{1,k}.
\end{align}
Since the $M_i$'s are projectors, $M_i^2 = M_i$. Define
\begin{align}
\epsilon^i_t = \bra{\psi^i_t}(I - M_i) \ket{\psi^i_t},\quad i=0,1,\cdots,N. \label{eq:epsilon AND}
\end{align}
Note that $\braket{\psi^i_{t+1}}{\psi^i_{t+1}} = \bra{\psi^i_t}M_i^2\ket{\psi^i_t} = \bra{\psi^i_t}M_i\ket{\psi^i_t} = \braket{\psi^i_t}{\psi^i_t} - \epsilon^i_t$, for all $i=0,\cdots,N$ (including 0!), and hence
\begin{align}
\sum\limits_{t=0}^{T-1} \epsilon^i_t = \braket{\psi^i_0}{\psi^i_0} - \braket{\psi^i_T}{\psi^i_T} \le \epsilon.
\end{align}
We now define the progress function. Let
\begin{align}
W^k_t = \braket{\psi^0_t}{\psi^k_t}
\end{align}
and let the progress function be a sum over $W^k$'s:
\begin{align}
W_t = \sum\limits_{k=1}^N W^k_t = \sum\limits_{k=1}^N\braket{\psi^0_t}{\psi^k_t}.
\end{align}

We can lower bound the total change in the progress function by (see \cite{ambainis-2000} for a proof; their proof equally applies to unnormalized states)
\begin{align}
W_0 - W_T \ge (1-2\sqrt{\delta(1-\delta)}) N. \label{eq:total AND}
\end{align}
We now proceed to upper bound $W_0 - W_T$. Note that
\begin{align}
W^k_t - W^k_{t+1} &= \braket{\psi^0_t}{\psi^k_t} - \bra{\psi^0_t}M_0M_k\ket{\psi^k_t} \nn \\
&= \bra{\psi^0_t}(I-M_0)M_k\ket{\psi^k_t} + \bra{\psi^0_t}M_0(I-M_k)\ket{\psi^k_t} + \bra{\psi^0_t}(I-M_0)(I-M_k)\ket{\psi^k_t}
\end{align}
and since $(I-M_0)M_k\ = 0$, $M_0(I-M_k) = \ket{1,k}\bra{1,k}$, we have
\begin{align}
W^k_t - W^k_{t+1} &\le \braket{\psi^0_t}{1,k}\braket{1,k}{\psi^k_t} + \L\| (I-M_0)\ket{\psi^0_t} \R\| \L\| (I-M_k)\ket{\psi^k_t} \R\|  \nn \\
&\le \L \| \braket{1,k}{\psi^0_t}\R \| + \sqrt{{\epsilon^0_t}{\epsilon^k_t}}. 
\end{align}
where we used \ref{eq:epsilon AND}. Summing over $k$ and t, we obtain
\begin{align}
W_0 - W_T &= \sum\limits_{t=0}^{T-1}\sum\limits_{k=1}^{N} \L[ \L \| \braket{1,k}{\psi^0_t}\R \| + \sqrt{{\epsilon^0_t}{\epsilon^k_t}} \R] \nn \\
&\le \sqrt{TN} \sqrt{\sum\limits_{t=0}^{T-1}\sum\limits_{k=1}^{N} \braket{\psi^0_t}{1,k}\braket{1,k}{\psi^0_t}} + \sum\limits_{t=0}^{T-1}\sum\limits_{k=1}^{N} \frac{\epsilon^0_t + \epsilon^k_t}{2} \nn \\
&\le \sqrt{TN} \sqrt{\sum\limits_{t=0}^{T-1}\bra{\psi^0_t}(I-M_0)\ket{\psi^0_t}} +  N\epsilon \nn \\
&\le \sqrt{TN\sum\limits_{t=0}^{T-1}\epsilon^0_t} + N\epsilon \nn \\ 
&\le \sqrt{\epsilon TN} + N\epsilon
\end{align}
where in the second line we used Cauchy-Schwarz and the AM-GM inequality.
Combined with $W_0 - W_T \ge (1-2\sqrt{\delta(1-\delta)})N$ (Eq. \ref{eq:total AND}), this immediately gives us
\begin{align}
T \ge \frac{(1-2\sqrt{\delta(1-\delta)}-\epsilon)^2N}{\epsilon}.
\end{align}
\end{proof}

We now proceed to prove the general result. This proof follows the presentation given in A. Childs's online lecture notes \cite{childs-2013}, which we found quite illuminating.
\begin{repthm}{thm:lowerbound}
For all functions $f$ with boolean input alphabet, and numbers $\epsilon$, $\delta$ satisfying $0 < \epsilon \le \delta \le 1/10$, 
\begin{equation}
B_{\epsilon,\delta}(f) = \Omega(Q_{0.01}(f)^2/\epsilon).
\end{equation}
\end{repthm}

\begin{proof}
We prove the lower bound on $B_{\epsilon,\delta}$ by showing that it is lower bounded by $\Omega(\text{Adv}^\pm(f)^2/\epsilon)$, where $\text{Adv}^\pm(f)$ is the generalized (i.e. allowing negative weights) adversary bound \cite{hoyer-2007} for $f$. We can then derive our theorem from the result \cite{lee-2010} that $Q(f) = O(\text{Adv}^\pm(f))$.

We generalize the bound on the $f=AND$ case to an adversary bound for $B_{\epsilon,\delta}$ on arbitrary $f$. Define the projectors
\begin{align}
\Pi_0 &= \sum\limits_{i =1}^N \ket{0,i}\bra{0,i} \nn \\
\Pi_i &= \ket{1,i}\bra{1,i},\quad i=1,\cdots,n.
\end{align}
It is clear that
\begin{align}
\Pi_0 + \sum\limits_{i=1}^N \Pi_i = I.
\end{align}
Note that $M_x = CP_{x,0}$ is
\begin{align}
M_x = \Pi_0 + \sum\limits_{i:x_i=0} \Pi_i.
\end{align}

Define $\ket{\psi^x_t}$ as the state of the algorithm right before the $(t+1)$-th query with input $x$; then
\begin{align}
\ket{\psi^x_{t+1}} = U_{t+1}M_x\ket{\psi^x_t}
\end{align}
for some unitary $U_{t+1}$. Now if we let
\begin{align}
\epsilon^x_t = \bra{\psi^x_t}(I-M_x)\ket{\psi^x_t}
\end{align}
then it follows that $\braket{\psi^x_t}{\psi^x_t} - \braket{\psi^x_{t+1}}{\psi^x_{t+1}} = \epsilon^x_t$, and thus
\begin{align}
\sum\limits_{t=0}^{T-1} \epsilon^x_t = \braket{\psi^x_0}{\psi^x_0} - \braket{\psi^x_T}{\psi^x_T} \le \epsilon.
\end{align}

We proceed to define the progress function. Let $S$ be the set of allowable input strings $x$. Let $\Gamma$ be an \emph{adversary matrix}, i.e. an $S \times S$ matrix such that
\begin{enumerate}
\item $\Gamma_{xy} = \Gamma_{yx} \quad \forall x,y \in S$; and
\item $\Gamma_{xy} = 0 \quad \text{if } f(x) = f(y)$.
\end{enumerate}
Let $a$ be the normalized eigenvector of $\Gamma$ with eigenvalue $\pm \|\Gamma\|$, where $\pm\|\Gamma\|$ is the largest (by absolute value) eigenvalue of $\Gamma$. Define the progress function
\begin{equation}
W_t = \sum\limits_{x,y \in S} \Gamma_{xy}a_x^*a_y \braket{\psi^x_t}{\psi^y_t}.
\end{equation}
For $\epsilon \le \delta < 1/10$ we have that\footnote{As described in \cite{hoyer-2007}, the $2\delta$ term can be removed if the output is boolean (0 or 1).\label{foot:2delta}} (see \cite{hoyer-2007} for a proof; their proof applies equally well to unnormalized states)
\begin{align}
|W_0 - W_T| \ge (1-2\sqrt{\delta(1-\delta)}-2\delta) \| \Gamma \|
\end{align}

We now proceed to upper bound $|W_0 - W_T| \le \sum_t |W_t - W_{t-1}|$. Note that
\begin{align}  \label{eq:threeterms}
W_{t} - W_{t+1} &= \sum\limits_{x,y \in S} \Gamma_{xy}a_x^*a_y\L(\braket{\psi^x_t}{\psi^y_t} - \braket{\psi^x_{t+1}}{\psi^y_{t+1}}\R) \nn \\
&= \sum\limits_{x,y \in S} \Gamma_{xy}a_x^*a_y\L(\braket{\psi^x_t}{\psi^y_t} - \bra{\psi^x_t}M_xM_y\ket{\psi^y_t}\R) \nn \\
&= \sum\limits_{x,y \in S} \Gamma_{xy}a_x^*a_y\L(\bra{\psi^x_t}(I-M_x)M_y\ket{\psi^y_t} + \bra{\psi^x_t}M_x(I-M_y)\ket{\psi^y_t} + \bra{\psi^x_t}(I-M_x)(I-M_y)\ket{\psi^y_t}\R)
\end{align}
We bound the three terms separately. For the first two terms, use
\begin{align}
(I - M_x) M_y &= \sum\limits_{i:x_i = 1, y_i=0} \Pi_i \nn \\
&= (I-M_x) \sum\limits_{i: x_i \neq y_i} \Pi_i
\end{align}
Define the $S \times S$ matrix $\Gamma_i$ as
\begin{align}
\Gamma_i = \begin{cases} \Gamma_{xy} &\mbox{if } x_i \neq y_i \\ 0 &\mbox{if } x_i = y_i \end{cases}
\end{align}
The first term of \ref{eq:threeterms} is
\begin{align}
\sum\limits_{x,y \in S}\sum\limits_{i:x_i\neq y_i} \Gamma_{xy}a_x^*a_y \bra{\psi^x_t}(I-M_x)\Pi_i\ket{\psi^y_t}
&= \sum\limits_{x,y \in S} \sum\limits_{i=1}^N \L(\Gamma_i\R)_{xy}a_x^*a_y \bra{\psi^x_t}(I-M_x)\Pi_i\ket{\psi^y_t} \nn \\
&= \sum_{i=1}^N \tr(Q_i \Gamma_i \tilde{Q}_i^\dagger)
\end{align}
where
\begin{align}
Q_i &= \sum\limits_{x \in S} a_x \Pi_i \ket{\psi^x_t}\bra{x} \\
\tilde{Q}_i &= \sum\limits_{x \in S} a_x \Pi_i (I-M_x)\ket{\psi^x_t}\bra{x}.
\end{align}
Although both $Q_i$ and $\tilde{Q}_i$ depend on $t$, we suppress the $t$ dependence in the notation.
Similarly, the second term of \ref{eq:threeterms} is equal to $\sum_{i=1}^N \tr(\tilde{Q}_i \Gamma_i Q_i^\dagger)$. 
We can also rewrite the third term of  \ref{eq:threeterms} as
\begin{align}
 \sum\limits_{x,y \in S} \Gamma_{xy}a_x^*a_y \bra{\psi^x_t}(I-M_x)(I-M_y)\ket{\psi^y_t} =\tr (Q' \Gamma Q'^\dag)
\end{align}
where 
\begin{align}
Q' &= \sum\limits_{x \in S} a_x (I-M_x)\ket{\psi^x_t}\bra{x}.
\end{align}
Therefore, adding absolute values,
\begin{align} \label{eq:beforelemma}
|W_{t} - W_{t+1}| &\le \sum_{i=1}^N \L[\L|\tr(Q_i \Gamma_i \tilde{Q}_i^\dagger)\R| + \L|\tr(\tilde{Q}_i \Gamma_i Q_i^\dagger))\R|\R] + \L|\tr (Q' \Gamma Q'^\dag)\R|
\end{align}

To continue, we need the following lemma:
\begin{lem}\label{lem:norms}
For any $m,n > 0$ and matrices $X \in \mathbb{C}^{m\times n}$, $Y \in \mathbb{C}^{n \times n}$, $Z \in \mathbb{C}^{n \times m}$, we have $|\tr(XYZ)| \le \|X\|_F \|Y\| \|Z\|_F$. Here $\| \cdot \|$ and $\| \cdot \|_F$ denote the spectral norm and Frobenius norm, respectively.
\end{lem}
This lemma can be proved by using that $|\tr(XYZ)| \le \|Y\| \|XZ\|_{tr}$ and $\|XZ\|_{tr} \le \|X\|_F \|Z\|_F$, which follows from \cite[Exercise IV.2.12 and Corollary IV.2.6]{bhatia-1997}. A more accessible proof is found online at \cite{childs-2013}.

Then by Lemma \ref{lem:norms},
\begin{align} \label{eq:trqgq1}
\sum_{i=1}^N \L|\tr(Q_i \Gamma_i \tilde{Q}_i^\dagger)\R| &\le \sum_{i=1}^N \|\Gamma_i\| \|Q_i\|_F \|\tilde{Q}_i\|_F
\end{align}
Since
\begin{align}
\sum\limits_{i=1}^N \|Q_i\|_F^2 &= \sum\limits_{i=1}^N\sum\limits_{x \in S} |a_x|^2 \L\| \Pi_i \ket{\psi^x_t} \R\|^2 \nn \\
&= \sum_{x\in S} |a_x|^2 \bra{\psi^x_t} \sum\limits_{i=1}^N \Pi_i \ket{\psi^x_t} \nn \\
&\le \sum_{x \in S} |a_x|^2 \nn \\
&= 1
\end{align}
and
\begin{align}
\sum\limits_{i=1}^N \|\tilde{Q}_i\|_F^2 &= \sum\limits_{i=1}^N\sum\limits_{x \in S} |a_x|^2 \L\| \Pi_i (I-M_x)\ket{\psi^x_t} \R\|^2 \nn \\
&= \sum_{x\in S} |a_x|^2 \bra{\psi^x_t} (I-M_x) \L( \sum\limits_{i=1}^N \Pi_i \R) (I-M_x)\ket{\psi^x_t} \nn \\
&\le \sum_{x \in S} |a_x|^2 \bra{\psi^x_t}(I-M_x)\ket{\psi^x_t} \nn \\
&= \sum_{x \in S} |a_x|^2 \epsilon^x_t
\end{align}
we have, by Cauchy-Schwarz,
\begin{align}
\sum\limits_{i=1}^N \|Q_i\|_F \|\tilde{Q}_i\|_F \le \sqrt{\sum_{x \in S} |a_x|^2 \epsilon^x_t} \label{eq:cauchy schwarz adversary}
\end{align}
Therefore by \ref{eq:trqgq1} and \ref{eq:cauchy schwarz adversary},
\begin{align}
\sum_{i=1}^N \L|\tr(Q_i \Gamma_i \tilde{Q}_i^\dagger)\R| &\le \sqrt{\sum_{x \in S} |a_x|^2 \epsilon^x_t} \max_{i \in [N]} \| \Gamma_i \|.
\end{align}
Similartly for $\tr(Q' \Gamma Q'^\dag)$, we have
\begin{align}
 \|Q'\|_F^2 &= \sum\limits_{x \in S} |a_x|^2 \L\|  (I-M_x)\ket{\psi^x_t} \R\|^2 \nn \\
&= \sum_{x \in S} |a_x|^2 \bra{\psi^x_t}(I-M_x)\ket{\psi^x_t} \nn \\
&= \sum_{x \in S} |a_x|^2 \epsilon^x_t
\end{align}
and using Lemma \ref{lem:norms},
\begin{align}
\tr(Q' \Gamma Q'^\dagger) &\le \| Q' \|_F^2 \| \Gamma \| \\
&= \sum_{x\in S} |a_x|^2 \epsilon^x_t \|\Gamma\|
\end{align}
Thus continuing from \ref{eq:beforelemma}, we have that
\begin{align}
|W_{t} - W_{t+1}| &\le 2\sqrt{\sum_{x \in S} |a_x|^2 \epsilon^x_t} \max_{i \in [N]} \| \Gamma_i \| +\sum_{x \in S} |a_x|^2 \epsilon^x_t\| \Gamma \| 
\end{align}
Finally, if we sum the above over $t$ we obtain
\begin{align}
| W_0 - W_T | &\le 2 \max_{i \in [N]} \| \Gamma_i \| \sum\limits_{t=0}^{T-1} \sqrt{\sum_{x \in S} |a_x|^2 \epsilon^x_t}   + \sum\limits_{t=0}^{T-1}\sum_{x \in S} |a_x|^2 \epsilon^x_t\| \Gamma \| 
\end{align}
The first term can be bounded using Cauchy-Schwarz:
\begin{align}
\sum\limits_{t=0}^{T-1} \sqrt{\sum_{x \in S} |a_x|^2 \epsilon^x_t} &\le \sqrt{T\sum\limits_{t=0}^{T-1}\sum_{x \in S} |a_x|^2 \epsilon^x_t} \nn \\
&\le \sqrt{\epsilon T}
\end{align}
where we used $\sum_t \epsilon^x_t \le \epsilon$ and $\sum_x |a_x|^2=1$. The second term can be summed easily:
\begin{align}
 \sum\limits_{t=0}^{T-1}\sum_{x \in S} |a_x|^2 \epsilon^x_t\| \Gamma \|  & \le  \sum_{x \in S} |a_x|^2 \epsilon \| \Gamma \| \nn \\
&= \epsilon \|\Gamma\|.
\end{align}
Therefore
\begin{align}
| W_0 - W_T | &\le 2\sqrt{\epsilon T} \max_{i \in [N]} \| \Gamma_i \| + \epsilon \|\Gamma\|.
\end{align}
Combined with our lower bound $| W_0-W_T | \ge (1-2\sqrt{\delta(1-\delta)}-2\delta) \|\Gamma\|$, this immediately gives
\begin{align}
T \ge \frac{(1-2\sqrt{\delta(1-\delta)}-2\delta-\epsilon)^2 }{4\epsilon} \frac{\|\Gamma\|^2}{\max_{i \in [N]} \| \Gamma_i \|^2}.
\end{align}
Recalling that \cite{hoyer-2007}
\begin{align}
\text{Adv}^\pm(f) = \max_{\Gamma} \frac{\|\Gamma\|}{\max_{i \in [N]} \| \Gamma_i \|},
\end{align}
we obtain\footnote{For boolean output (0 or 1) the $2\delta$ term can be dropped, as we previously noted (Footnote \ref{foot:2delta}).}
\begin{align}
T &\ge \frac{(1-2\sqrt{\delta(1-\delta)}-2\delta-\epsilon)^2 }{4\epsilon}\text{Adv}^\pm(f)^2.
\end{align}

We now use the tight characterization of the quantum query complexity by the general weight adversary bound:
\begin{thm} [{\cite[Theorem 1.1]{lee-2010}}]
Let $f:\: D \rightarrow E$, where $D \subseteq \{0,1\}^N$. Then $Q_{0.01}(f) = O(\text{\textnormal{Adv}}^\pm (f))$.
\end{thm}
Combined with our result above, we obtain
\begin{align}
B_{\epsilon,\delta}(f) = \Omega\L(\frac{Q_{0.01}(f)^2}{\epsilon}\R).
\end{align}
\end{proof}

\section{Proof of Theorem \ref{thm:classical}}  \label{app:nonconstructive}
We restate and prove Theorem \ref{thm:classical}:
\begin{repthm} {thm:classical}
Let $f:\:D \rightarrow E$, where $D \subseteq \{0,1\}^N$. Suppose there is a classical randomized query algorithm $\mathcal{A}$, that makes at most T queries, and evaluates $f$ with bounded error. Let the query results of $\mathcal{A}$ on random seed $s_{\mathcal{A}}$ be $x_{p_1},x_{p_2},\cdots,x_{p_{\tilde{T}(x)}}$, $\tilde{T}(x) \le T$, where $x$ is the hidden query string.

Suppose there is another (not necessarily time-efficient) randomized algorithm $\mathcal{G}$, with random seed $s_{\mathcal{G}}$, which takes as input $x_{p_1},\cdots,x_{p_{t-1}}$ and $s_\mathcal{A}$, and outputs a guess for the next query result $x_{p_{t}}$ of $\mathcal{A}$. Assume that $\mathcal{G}$ makes no more than an expected total of $G$ mistakes (for all inputs $x$). In other words,
\begin{equation}
\Expect_{s_{\mathcal{A}},s_{\mathcal{G}}} \L\{\sum_{t=1}^{\tilde{T}(x)}\L| \mathcal{G}(x_{p_1},\cdots,x_{p_{t-1}},s_{\mathcal{A}},s_{\mathcal{G}}) - x_{p_t} \R| \R\} \le G \quad \forall x.
\end{equation}
Note that $\mathcal{G}$ is given the random seed $s_{\mathcal{A}}$ of $\mathcal{A}$, so it can predict the next  query index of $\mathcal{A}$. \\
Then $B_{\epsilon}(f) = O(TG/\epsilon)$, and thus (by Theorem \ref{thm:main result}) $Q(f) = O(\sqrt{TG})$. 
\end{repthm}

\begin{proof}
For the purposes of this proof, we use the characterization of $B$ by the modified bomb construction given in section \ref{sect:generalize bomb}. This proof is substantially similar to that of theorem \ref{thm:upperbound}.

The following circuit finds $x_i$ with zero probability of explosion if $x_i=a$, and with an $O(1/L)$ probability of explosion if $x_i \neq a$ (in both cases the value of $x_i$ found by the circuit is always correct):

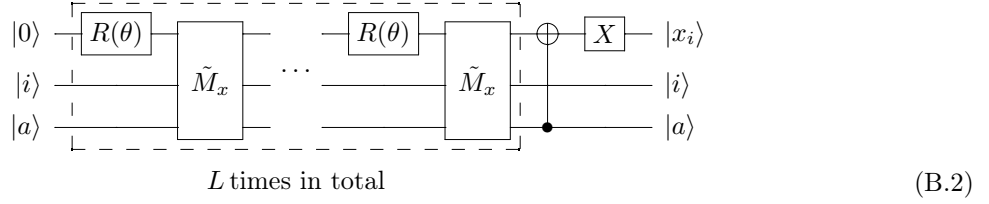
\begin{align}
\Qcircuit @C=1em @R=.7em {
\lstick{\ket{0}}   &\gate{R(\theta)}& \multigate{2}{\tilde{M}_x}  &\qw&    &&\gate{R(\theta)} & \multigate{2}{\tilde{M}_x} & \targ & \gate{X} 
&\rstick{\ket{x_i}}\qw   \\
\lstick{\ket{i}}  &\qw &\ghost{\tilde{M}_x}  &\qw    & \ustick{\dots} &    &\qw&\ghost{\tilde{M}_x}  & \qw & \qw &\rstick{\ket{i}}\qw\\ 
\lstick{\ket{a}} & \qw & \ghost{\tilde{M}_x}  & \qw  &                        &     &\qw & \ghost{\tilde{M}_x}  & \ctrl{-2} & \qw & \rstick{\ket{a}}\qw \\
&&&&\dstick{{L}\,\text{times in total}}  \gategroup{1}{2}{3}{8}{.7em}{--} \\
} \nn \\
 \, \label{eq:C2B}
\end{align} 
where $\theta = \pi /(2L)$ for some large number $L$ to be picked later, and
\begin{align}
R(\theta) =
\bpm
\cos\theta &-\sin \theta \\
\sin\theta &\cos\theta
\epm
\end{align}
The boxed part of the circuit is then simply $[\tilde{M}_x (R(\theta) \otimes I \otimes I)]^L$, applied to the state $\ket{0,i,a}$.
We can analyze this circuit by breaking into cases:

\begin{itemize}
\item If $x_i=a$, then $\tilde{M}_x\ket{\psi}\ket{i,a} = \ket{\psi}\ket{i,a} $ for any state $\ket{\psi}$ in the control register. Thus the $\tilde{M}_x$'s act as identities, and the circuit simply applies the rotation $R(\theta)^L=R(\pi/2)$ to the control register, rotating it from 0 to 1. We thus obtain the state $\ket{1,i,a}$; the final CNOT and X gates add $a\oplus 1 = x_i \oplus 1$ to the first register, giving $\ket{x_i,i,a}$.
\item If $x_i \neq a$, then
\begin{equation}
\tilde{M}_{x}\ket{0,i,a} = \ket{0,i,a}, \quad \tilde{M}_{x}\ket{1,i,a} = 0 \quad \text{(for } x_i \neq a \text{)}
\end{equation}
Therefore after each rotation $R(\theta)$, the projection $\tilde{M}_x$ projects the control qubit back to 0:
\begin{equation}
\tilde{M}_x(R(\theta)\otimes I \otimes I)\ket{0,i,a} = \tilde{M}_x(\cos\theta\ket{0} + \sin\theta\ket{1})\ket{i,a} = \cos\theta\ket{0,i,a} \quad \text{(for } x_i \neq a \text{)}
\end{equation}
In this case the effect of $\tilde{M}_x(R(\theta)\otimes I \otimes I)$ is to shrink the amplitude by $\cos(\theta)$; $L$ applications results in the state $\cos^L(\theta)\ket{0,i,a}$. The final CNOT and X gates add $a\oplus 1 = x_i$ to the first register, giving $\ket{x_i,i,a}$.
\end{itemize}
The probability of explosion is $0$ if $x_i = a$. If $x_i \neq a$, the probability of explosion is
\begin{equation}
1 - \cos^{2L}\left(\frac{\pi}{2L}\right) \le \frac{\pi^2}{4L}.
\end{equation}
Pick
\begin{equation}
L = \L\lceil \frac{\pi^2 G}{4\epsilon}\R\rceil.
\end{equation}
Then the probability of explosion is $0$ if $x_i=a$, and no more than $\epsilon/G$ if $x_i \neq a$. If the bomb does not explode, then the circuit \emph{always} finds the correct value of $x_i$.

We now construct the bomb query algorithm based on $\mathcal{A}$ and $\mathcal{G}$. The bomb query algorithm follows $\mathcal{A}$, with each classical query replaced by the above construction. There are no more than $TL\approx\pi^2TG/(4\epsilon)$ bomb queries. At each classical query, we pick the guess $a$ to be the guess provided by $\mathcal{G}$. The bomb only has a chance of exploding if the guess is incorrect; hence for all $x$, the total probability of explosion is no more than
\begin{equation}
\frac{\epsilon}{G}\: \Expect_{s_{\mathcal{A}},s_{\mathcal{G}}}  \L\{\sum_{t=1}^{\tilde{T}(x)}\L| \mathcal{G}(x_{p_1},\cdots,x_{p_{t-1}},s_{\mathcal{A}},s_{\mathcal{G}}) - x_{p_t} \R| \R\} \le \epsilon
\end{equation}
Thus replacing the classical queries of $\mathcal{A}$ with our construction gives a bomb query algorithm with probability of explosion no more than $\epsilon$; aside from the probability of explosion, this bomb algorithm makes no extra error over the classical algorithm $\mathcal{A}$. The number of queries this algorithm uses is
\begin{equation}
\tilde{B}_{\epsilon,\delta+\epsilon}(f) \le \L\lceil\frac{\pi^2G}{4\epsilon}\R\rceil T,
\end{equation}
where $\delta$ is the error rate of the classical algorithm. Therefore by Lemma \ref{lem:equivalent} and Lemma \ref{lem:delta}, 
\begin{align}
{B}_{\epsilon}(f) =O({B}_{\epsilon,\delta+\epsilon}(f))= O(\tilde{B}_{\epsilon,\delta+\epsilon}(f) )=O\L( {TG}/{\epsilon}\R)
\end{align}
\end{proof}
\section{Proof of Theorem \ref{thm:first marked element}} \label{app:first marked element}
We restate and prove Theorem \ref{thm:first marked element}:
\begin{repthm}{thm:first marked element}[Finding the first marked element in a list]
Suppose there is an ordered list of $N$ elements, and each element is either marked or unmarked. Then there is a bounded-error quantum algorithm for finding the \textbf{first} marked element in the list, or determines that no marked elements exist, such that:
\begin{itemize}
\item If the first marked element is the $d$-th element of the list, then the algorithm uses an expected $O(\sqrt{d})$ time and queries.
\item If there are no marked elements, then the algorithm uses $O(\sqrt{N})$ time and queries.
\end{itemize}
\end{repthm}
\begin{proof}
We give an algorithm that has the stated properties. We first recall a quantum algorithm for finding the minimum in a list of items:

\begin{thm}[{\cite{durr-1996}}] \label{thm:minimum}
Given a function $g$ on a domain of $N$ elements, there is a quantum algorithm that finds the minimum of $g$ with expected $O(\sqrt{N})$ time and evaluations of $g$, making $\delta < 1/10$ error.
\end{thm}

We now give our algorithm for finding the first marked element in a list. For simplicity, assume that $N$ is a power of 2 (i.e. $\log_2N$ is an integer).

\begin{alg} \label{alg:first marked element}  \leavevmode

\begin{enumerate}
\item For $\ell = 2^0, 2^1, 2^2, \cdots, 2^{\log_2 N} = N$:
\begin{itemize}
\item Find the first marked element within the first $\ell$ elements, or determine no marked element exists. This can be done by defining
\begin{equation}
g(i) = \begin{cases} \infty &\mbox{if } i \text{ is unmarked} \\ i &\mbox{if } i \text{ is marked}, \end{cases}
\end{equation}
and using Theorem \ref{thm:minimum} to find the minimum of $g$.  This takes $O(\sqrt{\ell}) = O(\sqrt{d})$ queries and makes $\delta < 1/10$ error for each $\ell$. If a marked element $i^*$ is found, the algorithm outputs $i^*$ and stops.
\end{itemize}
\item If no marked element was found in Step 1, the algorithm decides that no marked element exists.
\end{enumerate}
\end{alg}

We now claim that Algorithm \ref{alg:first marked element} has the desired properties. Let us break into cases:

\begin{itemize}
\item If no marked items exist, then no marked item can possibly be found in Step 1, so the algorithm correctly determines that no marked items exist in Step 2. The number of queries used is 
\begin{equation}
\sum_{i=0}^{\log_2N} \sqrt{2^i} = O(\sqrt{N})
\end{equation}
as desired.
\item Suppose the first marked item is the $d$-th item in the list. Then in Step 1(a), if $\ell \ge d$, there is at least a $1-\delta$ probability that the algorithm will detect that a marked item exists in the first $\ell$ elements and stop the loop. Letting $\alpha = \lceil \log_2 d \rceil$, the total expected number of queries is thus
\begin{align}
\sum_{i=0}^{\alpha-1} \sqrt{2^i } + \sum_{i=\alpha}^{\log_2N} \delta^{i-\alpha}\sqrt{2^i } + O(\sqrt{d}) &\le \frac{2^{\alpha/2}-1}{\sqrt{2}-1} +\sqrt{2^{\alpha}} \frac{1}{1-\sqrt{2}\delta} + O(\sqrt{d}) \\
&= O(\sqrt{2^\alpha }) + O(\sqrt{d}) \\
&= O(\sqrt{d}).
\end{align}
The probability of not finding the marked item at the first $\ell \ge d$ is at most $\delta$ , and thus the total error of the algorithm is bounded by $\delta$.
\end{itemize}
\end{proof}

\end{document}